\documentclass{IEEEtran}
\usepackage{color} 
\usepackage{times,amssymb,amsmath,amsfonts,amsthm,bbm,mathrsfs,color,enumitem,booktabs,cite,bm,caption,subcaption}
\usepackage[normalem]{ulem}
\usepackage[mathscr]{eucal}
\allowdisplaybreaks
\usepackage{hyperref}
\newcommand\redout{\bgroup\markoverwith
{\textcolor{red}{\rule[0.5ex]{2pt}{0.8pt}}}\ULon}

\newtheorem{theorem}{Theorem}[section]
\newtheorem{lemma}[theorem]{Lemma}

\newtheorem{proposition}[theorem]{Proposition}

\newtheorem{corollary}[theorem]{Corollary}

\newcommand\nc\newcommand
\nc\ffa{{\boldsymbol a}}\nc\ffA{{\boldsymbol A}}\nc\cA{{\mathscr A}}
\nc\ffb{{\boldsymbol b}}\nc\ffB{{\boldsymbol B}}\nc\cB{{\mathscr B}}
\nc\ffc{{\boldsymbol c}}\nc\ffC{{\boldsymbol C}}\nc\cC{{\mathscr C}}
\nc\ffd{{\boldsymbol d}}\nc\ffD{{\boldsymbol D}}\nc\cD{{\mathscr D}}
\nc\ffe{{\boldsymbol e}}\nc\ffE{{\boldsymbol E}}\nc\cE{{\mathscr E}}
\nc\fff{{\boldsymbol f}}\nc\ffF{{\boldsymbol F}}\nc\cF{{\mathscr F}}
\nc\ffg{{\boldsymbol g}}\nc\ffG{{\boldsymbol G}}\nc\cG{{\mathscr G}}\nc\bG{{\mathbb G}}
\nc\ffh{{\boldsymbol h}}\nc\ffH{{\boldsymbol H}}\nc\cH{{\mathscr H}}
\nc\ffi{{\boldsymbol i}}\nc\ffI{{\boldsymbol I}}\nc\cI{{\mathcal I}}
\nc\ffj{{\boldsymbol j}}\nc\ffJ{{\boldsymbol J}}\nc\cJ{{\mathscr J}}
\nc\ffk{{\boldsymbol k}}\nc\ffK{{\boldsymbol K}}\nc\cK{{\mathscr K}}
\nc\ffl{{\boldsymbol l}}\nc\ffL{{\boldsymbol L}}\nc\cL{{\mathscr L}}
\nc\ffm{{\boldsymbol m}}\nc\ffM{{\boldsymbol M}}\nc{\cM}{{\mathscr M}}
\nc\ffn{{\boldsymbol n}}\nc\ffN{{\boldsymbol N}}\nc\cN{{\mathscr N}}
\nc\ffo{{\boldsymbol o}}\nc\ffO{{\boldsymbol O}}\nc\cO{{\mathscr O}}
\nc\ffp{{\boldsymbol p}}\nc\ffP{{\boldsymbol P}}\nc\cP{{\mathscr P}}\nc\eP{{\EuScriptP}}\nc\fP{{\mathfrak P}}
\nc\ffq{{\boldsymbol q}}\nc\ffQ{{\boldsymbol Q}}\nc\cQ{{\mathscr Q}}
\nc\ffr{{\boldsymbol r}}\nc\ffR{{\boldsymbol R}}\nc\cR{{\mathscr R}}
\nc\ffs{{\boldsymbol s}}\nc\ffS{{\boldsymbol S}}\nc\cS{{\mathscr S}}
\nc\fft{{\boldsymbol t}}\nc\ffT{{\boldsymbol T}}\nc\cT{{\mathscr T}}
\nc\ffu{{\boldsymbol u}}\nc\ffU{{\boldsymbol U}}\nc\cU{{\mathscr U}}
\nc\ffv{{\boldsymbol v}}\nc\ffV{{\boldsymbol V}}\nc\cV{{\mathscr V}}
\nc\ffw{{\boldsymbol w}}\nc\ffW{{\boldsymbol W}}\nc\cW{{\mathscr W}}
\nc\ffx{{\boldsymbol x}}\nc\ffX{{\boldsymbol X}}\nc\cX{{\mathscr X}}
\nc\ffy{{\boldsymbol y}}\nc\ffY{{\boldsymbol Y}}\nc\cY{{\mathscr Y}}
\nc\ffz{{\boldsymbol z}}\nc\ffZ{{\boldsymbol Z}}\nc\cZ{{\mathscr Z}}
\nc{\bb}{{\mathbbm{1}}}
\nc\reals{{\mathbb R}}
\nc{\ff}{{\mathbb F}}
\nc{\PP}{{\mathbb P}}

\DeclareMathOperator{\diam}{{\sf diam}}
\DeclareMathOperator{\diag}{diag}

\usepackage{tikz}
\usetikzlibrary{calc}
\usetikzlibrary{decorations.pathreplacing,decorations.markings,snakes,shapes.geometric}
\tikzset{snake it/.style={decorate, decoration=snake}}

\newcommand\remove[1]{}

\DeclareSymbolFont{bbold}{U}{bbold}{m}{n}
\DeclareSymbolFontAlphabet{\mathbbold}{bbold}

\begin{document}

\title{{Node repair on connected graphs}}

		\author{\IEEEauthorblockN{Adway Patra} \hspace*{1in}
\and \IEEEauthorblockN{Alexander Barg,}\IEEEmembership{\ Fellow, IEEE}}
		\date{}
		\maketitle
\maketitle

\maketitle

\begin{abstract}
We study the problem of erasure correction (node repair) for regenerating codes defined on graphs wherein the cost of transmitting the information to the failed node depends on the graphical distance from this node to the helper vertices of the graph. The information passed to the failed node from the helpers traverses several vertices of the graph, and savings in communication complexity can be attained if the intermediate vertices process the information rather than simply relaying it toward the failed node. 
We derive simple information-theoretic bounds on the amount of information communicated between the nodes in the course of the repair. 
Next we show that Minimum Storage Regenerating (MSR) codes can be modified to perform the intermediate processing, thereby attaining the
lower bound on the information exchange on the graph. We also consider node repair when the underlying graph is random, deriving 
conditions on the parameters that support recovery of the failed node with communication complexity smaller than required by the simple relaying.

\end{abstract}

\renewcommand{\thefootnote}{\arabic{footnote}}
\setcounter{footnote}{0}		
		
{\renewcommand{\thefootnote}{}\footnotetext{
\vspace*{-.15in}

\noindent\rule{1.5in}{.4pt}

{Results of this paper were presented in part at the 2021 IEEE International Symposium on Information Theory.

The authors are with Dept. of ECE and ISR, University of Maryland, College Park, MD 20742. Emails: \{apatra, abarg\}@umd.edu.
Alexander Barg is also with IITP, Russian Academy of Sciences, 127051 Moscow, Russia.

This research was supported by NSF grants CCF1814487, CCF2110113, and CCF2104489.
}}}
\renewcommand{\thefootnote}{\arabic{footnote}}
\setcounter{footnote}{0}

\section{Introduction}
Applications of erasure-correcting codes in distributed storage are focused on recovering a single erasure under the constraint on the total
amount of data ``moved'' from the other coordinates to correct the erased (failed) coordinate. This processing is commonly modeled 
by assuming that the codeword coordinates are placed on different servers (storage nodes), and aims at limiting the information communicated between them for the recovery of the failed node. The currently adopted framework for studying erasure correction under communication constraints was established in \cite{Dimakis10}, and coding constructions that minimize the communication are known collectively as {\em regenerating codes}. The authors of \cite{Dimakis10} derived a lower bound on the minimum amount of information acquired from the surviving nodes for the purposes of repair. 

For a finite field $\ff=\ff_q$ we consider a code $\cC\subset \ff^{nl}$ whose codewords are represented by $l\times n$ matrices over $\ff$.
We assume that each coordinate (a vector in $\ff^l$) is written on a single storage node, and that a failed node amounts to having its coordinate erased. The task of node repair can be thought of as correcting a single erasure in the vector code of length $n$ over $\ff.$ In this paper we address communication complexity of node repair under the assumption that communication between the nodes is constrained by 
a (connected) graph $G(V,E)$, where $V$ is an $n$-set of vertices and the cost of sending a unit of information from $v_i$ to $v_j$ is 
determined by the graph distance $\rho(v_i,v_j)$ in $G$. This group of problems is motivated by the assumption 
that the links between the nodes are established based on physical proximity and the associated energy constraints, limitations of the system architecture, or other features with the same effect. In the network environment such as low-power wide-area networks (LP-WAN), e.g., path loss in narrow-band lower-power IoT, the mentioned limitations arise naturally as a part of the functioning of the system.
We also always assume point-to-point rather than broadcast communication. For distributed storage systems this is a natural restriction, while for IoT applications this assumption may be imposed because of energy or privacy considerations.

Under a naive approach to this problem, it is still possible to use the known methods of node repair whereby the chosen group of the helper nodes communicates some functions of their contents to the failed node. Note however that the data from the helper nodes not directly connected to the failed node will have to be relayed along some path to the failed node, increasing the bandwidth utilized for the repair. Thus, a natural question to study is whether there are more economical ways of accomplishing this goal given the structure of the graph $G,$ under which the data from the far-off  helper nodes is processed along the way and combined with the contents of the intermediate nodes, saving on the overall communication. We give an affirmative answer in Sec.~\ref{sec:MSR}, showing that if the data is encoded using a {\em minimum storage regenerating}, or MSR code, then under some conditions it is possible to save on the communication cost of node repair compared to simple relaying of the information. More precisely, we derive a lower bound on the repair bandwidth on a graph and show that this bound can be attained with linear MSR codes. We note that the problem of repair on graphs was considered in earlier \cite{Gerami2014} in the case 
when the graph is a path and in \cite{LuGuangFu2014} for a directed ring. These papers derived lower bounds on the repair bandwidth which coincide with our bound once it is specialized to a path or a ring. 

Intermediate data processing is also an essential component of a version of the node repair problem known as {\em cooperative repair} 
\cite{ShumHu2013}. It is therefore of interest to examine possible applications of cooperative MSR codes to the problem at hand, aiming again at reduced communication complexity of repair. We show one application of this idea in Sec.~\ref{sec:multiple}, using 
a family of cooperative codes to design a scheme with reduced repair bandwidth in the case of multiple failed nodes.
In Sec.~\ref{sec:coop} we consider the problem of node repair in the situation when the helper nodes can exchange (and process) information
before communicating with the failed nodes. We derive a framework to bound below the complexity of repair under this relaxation
and use it to compute lower bounds on the repair bandwidth in several examples. One of these examples also affords a matching
code construction, again inspired by cooperative codes. Since the setting of this example is rather restricted, we give the details
in the appendix. Finally, in Sec.~\ref{sec:random} we address the question of node repair for random graphs from the standard Erd{\"o}s-R{\'e}nyi ensemble $\cG_{n,p}$ as well as from the ensemble of random regular graphs, and determine a range of parameters under which the communication cost of repair with intermediate processing is advantageous over the repair scheme based on the relaying. 

Throughout the paper we focus only on the node repair problem and do not study the communication complexity of the actual access to the encoded data (the ``data collection'' task in the terminology of \cite{Dimakis10}).

To conclude this introduction, we note that the general problem of information processing or recovery under communication constraints represented by a graph has recently been studied in a number of specific settings. Among them, locally recoverable codes on graphs \cite{Maz2015,MazMcgVor2019} (and the associated problems of guessing games on graphs \cite{Gadouleau2018} and index coding \cite{BKL2013}), their extension to recoverable systems \cite{ElishcoBarg2020}, private information retrieval on graphs \cite{Raviv2020}, and others. The problem of node repair under communication constraints introduced here is another instantiation of this broadly defined theme.

\section{Lower bounds on the repair bandwidth}\label{sec:LowerBound}

\subsection{MSR codes: A reminder}  Let $\ff$ be a finite field. A {\em vector code} $\cC$ of length $n$ is an $\ff$-linear subspace of $(\ff^{l})^{n}$ whose codewords can be thought of as $l\times n$ matrices. In the context of storage codes, elements of $\cC$ are often referred to as $n$-words whose coordinates are $l$-vectors over $\ff.$  We further assume that the information contents of the codeword is $kl$ symbols of $\ff$, in other words, that $|\cC|=q^{kl}$, and that any $k$ coordinates suffice to recover the entire codeword. 
Thus, the code has the maximum distance separable (MDS) property, and any $n-k$ erased coordinates can be found from the remaining $k$
ones, accounting for the optimal erasure correction capacity.

Suppose that the coordinates of the codeword are placed on $n$ different storage nodes, and refer to the coordinates themselves
as nodes. The defining property of MSR codes is related to recovering the value of an erased coordinate of the codeword, or repairing a 
single failed node. According to the above description, we can accomplish this by using the information from $k$ functional nodes
and downloading a total of $kl$ symbols of the field $\ff.$ At the same time, this operation supports recovery of the entire
codeword, accomplishing more than we actually need. An important finding of the work \cite{Dimakis10} was to point out that
we can save on the amount of downloaded information by performing the repair based on the contents of $d>k$ helper nodes. To achieve
the saving, each of the helper nodes provides a function of its contents, and \cite{Dimakis10} 
showed that to accomplish the repair it is necessary to download at least $\frac{dl}{d-k+1}$ field symbols. This is smaller than $kl$ for all $d\le n.$ A code $\cC$ with the parameters $(n,k,d,l)$ is called MSR if it supports node recovery with {\em repair bandwidth} meeting the lower bound for the chosen number $d$ of {\em helper nodes}. It is easy to show that for a code to have this property, each of the helper nodes necessarily provides $l/(d-k+1)$ field symbols for the recovery of the failed node (the so-called {\em uniform download} property).

Formally, an $(n,k,d,l)$ linear MDS vector code $\cC$ over $\ff$ is called MSR if there are linear functions $h_i:\ff^{l}\to \ff^{l/(d-k+1)},
i=1,\dots,n$ such that for any $j\in[n]$ and any subset $\{i_1,\dots,i_d\}\subset[n]\backslash \{j\}$ there exists a linear function $g_j:\ff^{d(l/(d-k+1))}\to\ff^l$ such that for any codeword $C=(C_1,\dots,C_n)\in \cC$ the value $C_j$ (the contents of the failed node) is found as
    $$
    C_j=g_j(h_{i_1}(C_{i_1}),\dots,h_{i_d}(C_{i_d})).
    $$
Slightly more generally, the functions $h_i$ could also depend on $j$, but this will not be important below. A number of families of MSR codes are known in the literature, among them constructions of \cite{Rashmi11,Tamo14,Raviv17,Goparaju17,Ye17,Ye16a}, see also a recent survey in \cite{KumarDistStorage2021}. In this paper we use two such families to exemplify our approach to node repair on graphs, namely {product matrix} codes \cite{Rashmi11} and {diagonal-matrix codes} \cite{Ye17}. It will become apparent toward the end of Sec.~\ref{sec:MSR} that any family of $\ff$-linear MSR codes
can be incorporated in our repair scheme. 

\subsection{Problem statement and the lower bounds} The problem that we consider is associating the nodes with the vertices of a graph and performing node repair by transmitting the
information along the edges. Let $\cC$ be an $(n,k,d,l)$ MSR code and suppose that each coordinate of a codeword $C\in \cC$ is written on a vertex of a graph $G(V,E)$ with $|V|=n,$ which represents a distributed storage system with a given connectivity structure. Suppose further that the coordinate $C_f, f\in[n]$ is erased, or, as we will say, that the node $v_f$ has failed. Let $D\subset V\backslash\{v_f\}, |D|=d$ be a set of helper nodes. To repair the failed node, the helper nodes provide information which is communicated to $v_f$ over the edges in $E$. If one discounts the connectivity constraints, then to accomplish the repair, each of the helper nodes sends the information to the failed node over the shortest path in $G$, and the intermediate nodes simply relay this information further, possibly supplementing it with their own data. We call this repair strategy {\em accumulate and forward} (AF).  To examine options for more economical repair including {\em intermediate processing} (IP) of the information, we begin with deriving a lower bound on the repair bandwidth.

Before proceeding, let us further specify our assumptions. We assume that for the failed node $v_f$, the helper nodes $D$ are chosen to be the $d$ \textit{closest} nodes to $v_f$ in terms of the graph distance\footnote{This assumption is not restrictive because, whenever
the set $D$ spans a connected subgraph, our bounds on communication complexity apply for the information processing within that subgraph.}. These nodes can be found by a simple breadth-first search on $G$ starting at $v_f$. Denote by $G_{f,D} = (V_{f,D},E_{f,D})$ the subgraph spanned by $\{v_f\}\cup D.$ Let $t= \max_{v \in D} \rho(v,v_f)$. We will use the following notation for spheres and balls around $v_f$ in $G_{f,D}:$
   $$\Gamma_j(v_f)=\{v\in V_{f,D}:\rho(v,v_f)=j\},\; N_i(v_f)=\cup_{j=1}^i \Gamma_j(v_f),
   $$ 
and we refer to the vertices in $\Gamma_{j}(v_f)$  as the helper nodes in \textit{layer} $j$. The case $t=1$ corresponds to the much-studied graph-agnostic repair scenario \cite{Dimakis10}, and therefore
we exclude it from consideration. Observe that the graph $G_{f,D}$ is not necessarily unique; in particular, there may be multiple possible choices for the helper nodes in the $t$-th layer. 

In the next lemma, we derive lower bounds on the amount of information contributed by a group of helper nodes for the purposes of repair. The lemma is phrased in information-theoretic terms. We assume that the information stored at the vertices is given by random variables $W_i,i\in[n]$ that have some joint distribution on $(\ff^l)^n$ and satisfy $H(W_i)=l$ for all $i$, where $H(\cdot)$ is the entropy. For a subset $A\subset V$ we write $W_A=\{W_i,i\in A\}.$
Let $S_i^f$ be the information provided to $v_f$ by the $i$th helper node in the traditional, fully connected repair setting, and let $S_D^f=\{S_i^f,i\in D\}.$ The RV $S_i^f$ is a function of the contents of the node $v_i$, and the RVs $S_i^f,i\in D$ determine the contents of $v_f,$  or formally,
   \begin{align*}
   &(i)\; H(S_i^f|W_i)=0, \quad i\in D,\\
   &(ii)\; H(W_f|S_D^f)=0.
   \end{align*}
From the cut-set bound \cite{Dimakis10} it follows that $H(S_i^f) \ge {l}/(d-k+1),$ and we assume that this is achieved with equality, i.e., the codes we use have the MSR property.
In the next lemma we bound below the amount of information sent by a subset of helper nodes in an MSR code. 
The proof that we give is close to the arguments that have previously appeared in the literature, see for instance \cite{Shah2012}.
\begin{lemma}\label{lemma:bound} Let $v_f, f\in[n]$ be the failed node.
For a subset of the helper nodes $A \subset D$ let $R_A^f$ be such that $H(R_A^f|W_A)=0$ and
\begin{equation}\label{eq:eqtn17}
		H(W_f|R_A^f, S_{D\backslash A}^f) = 0.
	\end{equation} 
1) If $|A| \ge d-k+1$, then
		$$
		H(R_A^f) \ge l.
		$$
2) If $|A| \le d-k$, then
		$$
		H(R_A^f) \ge \frac{|A|l}{d-k+1}.
		$$
\end{lemma}
\begin{proof}
Part (1): By the assumption \eqref{eq:eqtn17}, given the contents of all the nodes in $D\backslash A,$ the information contained in $R_A^f$ is sufficient to repair $v_f$, i.e., 
		\begin{equation}\label{eq:eqtn20}
			H(W_f|R_A^f, W_{D\backslash A})=0.
		\end{equation}
		We have $|D\backslash A| \le k-1$. Consider a set $B \subset A$ with $|B| = k-1-|D\backslash A|$. Now,
		\begin{equation}\label{eq:eqtn18}
			H(R_A^f, W_{D\backslash A}, W_{B}) = H(R_A^f, W_{D\backslash A}, W_f, W_{B}) \ge kl,
		\end{equation}
		where the equality in \eqref{eq:eqtn18} follows from \eqref{eq:eqtn20} and the chain rule, and the inequality follows from the MDS property of MSR codes because 
		$|D\backslash A|+|B|+1 = k$. Next observe that
		\begin{align}
			H(R_A^f, W_{ D\backslash A}, W_{B}) &\le H(R_A^f)+H( W_{D\backslash A}, W_{B})\notag\\ &= H(R_A^f)+ (k-1)l,\label{eq:eqtn19}
		\end{align}
		where the equality again uses the independence of any $k-1$ coordinates in an MDS code. Combining \eqref{eq:eqtn18} and \eqref{eq:eqtn19}, we obtain the claimed inequality.
	
 For Part (2), let $C \subseteq D\backslash A$ such that $|C| = k-1$ and let $I = D\backslash \{A\cup C\}$. By the assumption \eqref{eq:eqtn17}, we have
		\begin{equation}\label{eq:eqtn22}
			H(W_f|R_A^f,W_C,S_I^f) = 0.
		\end{equation}
		Now,
		\begin{equation}\label{eq:eqtn23}
			H(R_A^f,W_C,S_I^f) = H(R_A^f,W_f,W_C,S_I^f) \ge kl,
		\end{equation}
where the equality in \eqref{eq:eqtn23} follows from \eqref{eq:eqtn22} and the chain rule, and the inequality follows from the MDS property and the fact that $|C| = k-1$. Next observe that
		\begin{equation}\label{eq:eqtn24}
			\begin{split}
				&H(R_A^f,W_C,S_I^f)\\ &\le H(R_A^f)+H( W_C)+H(S_I^f)\\ 
				&\le H(R_A^f)+H( W_C)+\sum_{i \in D\backslash\{A\cup C\}}H(S_i^f)\\
				&= H(R_A^f)+ (k-1)l+\frac{(d-(k-1)-|A|)l}{d-k+1}
			\end{split}
		\end{equation}
		where we again use the independence of any $k-1$ coordinates in an MDS code. Combining \eqref{eq:eqtn23} and \eqref{eq:eqtn24}, we obtain the claimed inequality.
\end{proof}
Rephrasing this lemma, we obtain a lower bound on the amount of information transmitted between the layers in $G_{f,D}$.
\begin{proposition}\label{prop:layer}
Let $R_{j}^f$ be the random variable denoting the information flow from the $j$-th layer to the $(j-1)$-th layer.  Then 
$$ H(R_j^f) \ge \min\Big\{l, \frac{|\cup_{i=j}^t\Gamma_{i}(v_f)|\cdot l}{d-k+1}\Big\}$$
\end{proposition}
\noindent{\em Proof.}
	Follows from Lemma \ref{lemma:bound} by taking $A=\cup_{i=j}^t\Gamma_{i}(v_f)$.\qed

Note that $R_j^f$ in the above proposition represents the joint information transmitted by all the nodes in layer $j$ to layer $j-1$ 
and hence it does not account for any other communication occurring among the helper nodes.
If $G_{f,D}$ is a rooted tree, such communication does not occur, and this will be the main (but not the only) case studied below.
In this case we can make our arguments more precise.

\vspace*{.1in}

\begin{figure}[th]\begin{center}\scalebox{0.3}{\begin{tikzpicture}
				[
				vertex_style/.style={circle, draw, fill,minimum size=0.08cm,scale=1.25},
				vertex_style1/.style={circle, draw, fill=blue,minimum size=0.08cm,scale=1.25}
				]
				
				\useasboundingbox (-10,-10) rectangle (10,10);
				
				\begin{scope}[rotate=90]
					
					\node[circle,draw=red,fill=red,minimum size=0.1cm,scale=1.5,label={[xshift=1.2cm, yshift=-0.7cm,minimum size=0.5cm,color=red,scale=3]$v_f$}] (0) at (canvas polar cs: radius=0cm,angle=0){};
					
					\foreach \x/\y in {36/1,108/2,180/3,252/4,324/5}{
						\node[vertex_style] (\y) at (canvas polar cs: radius=4cm,angle=\x){};
					}
					\foreach \x/\y in {0/1,0/2,0/3,0/4,0/5}{
						\path[-] (\x) edge [black,ultra thick] (\y);
					}
				
					\foreach \x/\y in {13.5/6,28.5/7,43.5/8,58.5/9}{
						\node[vertex_style1] (\y) at (canvas polar cs: radius=8.5cm,angle=\x){};
						\path[-] (1) edge [blue,thick] (\y); 
					}
					\foreach \x/\y in {85.5/10,100.5/11,115.5/12,130.5/13}{
						\node[vertex_style1] (\y) at (canvas polar cs: radius=8.5cm,angle=\x){};
						\path[-] (2) edge [blue,thick] (\y); 
					}
					\foreach \x/\y in {85.5+72/14,100.5+72/15,115.5+72/16,130.5+72/17}{
						\node[vertex_style1] (\y) at (canvas polar cs: radius=8.5cm,angle=\x){};
						\path[-] (3) edge [blue,thick] (\y); 
					}
					\foreach \x/\y in {85.5+144/19,100.5+144/20,115.5+144/21,130.5+144/22}{
						\node[vertex_style1] (\y) at (canvas polar cs: radius=8.5cm,angle=\x){};
						\path[-] (4) edge [blue,thick] (\y); 
					}
					\foreach \x/\y in {85.5+216/23,100.5+216/24,115.5+216/25,130.5+216/26}{
						\node[vertex_style1] (\y) at (canvas polar cs: radius=8.5cm,angle=\x){};
						\path[-] (5) edge [blue,thick] (\y); 
					}
					\draw[black,ultra thick,dashed,label=x] (0,0) circle (4.0cm) ;
					\draw[blue,ultra thick,dashed] (0,0) circle (8.5cm);
					\node[minimum size=0.5cm,scale=3] at (1.75,-5.5) {$\Gamma_{1}(v_f)$};
					\node[minimum size=0.5cm,scale=3,color=blue] at (5.3,-9) {$\Gamma_{2}(v_f)$};
					\end{scope}
		\end{tikzpicture}}
		\centerline{\footnotesize {\sc Fig.1:} Repair tree of the node $v_f$}
		
	\end{center}
\end{figure}

Let $T_f$ be a rooted spanning tree of $G_{f,D}$ with root $v_f$ (see Fig.1), then it defines the set of descendants of each node in $T_f.$ Let $D(v_i)$ be the set of descendants of $v_i,$ and let $D^\ast(v_i)=D(v_i)\cup\{v_i\}.$ 
The total communication complexity of node repair using the tree $T_f$ is bounded below in the following proposition.
\begin{proposition}
Let $J_f = \{v \in V(T_f)\backslash\{v_f\}: |D^*(v)| \ge d-k+2\}$. The total communication complexity $\beta$ for the repair of node $v_f$ on the repair tree $T_f$ is bounded as
\begin{equation}\label{eq:eqtn25}
\beta \ge |J_f|l+\sum_{v \in V(T_f)\backslash(\{v_f\}\cup J_f)}\frac{|D^*(v)|l}{d-k+1}.
\end{equation} 
\end{proposition}
\begin{proof}
For every non-root node $v \notin J_f$, we have $|D^*(v)| \le d-k$. Since $T_f$ is a tree, any outflow of information 
out of the subtree spanned by $D^*(v)$ passes through the node $v$, so it needs to transmit at least $|D^*(v)|\cdot l/(d-k+1)$ symbols to 
its immediate parent in $T_f$ by Lemma \ref{lemma:bound}. By the same lemma, every node $v \in J_f$ 
needs to transmit at least $l$ symbols to its immediate parent.
\end{proof}

For comparison purposes we also write out an expression for the AF repair procedure of MSR codes, described in the beginning of this section. Its repair bandwidth can be found as
\begin{equation}\label{eq:bAF}
	\beta_{\text{\rm AF}} = \Big( t(d-|N_{t-1}(v_f)|)+\sum_{i=1}^{t-1}i|\Gamma_i(v_f)| \Big)\frac{l}{d-k+1}.
\end{equation}
Every helper node provides $l/(d-k+1)$ symbols of information for repair, so for a node $v\not\in J_f$ 
the AF strategy is trivially optimal by part (2) of Lemma \ref{lemma:bound}. At the same time, for nodes $v \in J_f$ a better communication strategy is not a priori ruled out. This problem is addressed in the next section.

\subsection{A bound for repair of multiple nodes}\label{sec:multiple_bound}
Before proceeding further, let us note a simple extension of Lemma \ref{lemma:bound} to the case of multiple failed nodes, which 
is often studied for regenerating codes under full node connectivity \cite{Ye17}. The repair of multiple nodes in a graph
depends on their mutual placement and their connections to the helpers, and gives rise to several options. Denote by $F\subset V$ 
the set of failed nodes, and let $|F|=h\ge 1$. To keep the argument manageable, we assume that recovery of all the nodes in $F$ relies on a {\em common} set $D$ of helper nodes. With this assumption Lemma \ref{lemma:bound} affords the following extension.

\begin{lemma}\label{lemma:bound_mult} Let $F \subset [n],|F| = h, 1\le h \le n-d$ be the set of failed nodes.
	For a subset of the helper nodes $A \subset D$ let $R_A^F$ be a function of $S_A^F$ such that 
	\begin{equation}\label{eq:eqtn_mult1}
		H(W_F|R_A^F, S_{D\backslash A}^F) = 0.
	\end{equation} 
	1) If $|A| \ge d-k+h$, then
	$$
	H(R_A^F) \ge hl.
	$$
	2) If $|A| \le d-k+h-1$, then
	$$
	H(R_A^F) \ge \frac{h|A|l}{d-k+h}.
	$$
\end{lemma}
The proof follows closely the proof of Lemma~\ref{lemma:bound} and will be omitted (it is included in the preprint version of this paper, arXiv:2108.00939, Appendix B). As above, in Lemma \ref{lemma:bound_mult} we sidestepped the specific way of communicating the information from the helpers to the failed nodes, limiting ourselves to the lower bounds on the information provided by the helpers. The communication complexity of implementing the repair depends on the topology of the graph and on the relative location of the failed nodes and the helpers.
In Sec.~\ref{sec:multiple} we present a construction that, under certain assumptions, attains the bounds of this lemma, performing the intermediate processing instead of relaying and gaining in communication complexity over the AF protocol.

%
%

\section{MSR constructions for repair of graph vertices}\label{sec:MSR}

In this section we show that linear MSR codes support a repair procedure that attains the lower bound \eqref{eq:eqtn25} on the 
communication complexity. While this procedure is general, we begin with illustrating it for product-matrix codes of \cite{Rashmi11}.
Then we consider several examples of graphs, estimating the savings of repair complexity compared to the AF repair. After that,
we show that conceptually the same procedure applies to the diagonal-matrix codes of \cite{Ye17}, and then briefly discuss
a general version of this repair protocol as it applies to all families of $\ff$-linear MSR codes.

\subsection{Product-matrix (PM) codes} In their standard form, PM codes are described as follows. Fix the code length $n$ and the dimension parameter $k,$ and take $d=2k-2,l=k-1.$
The code $\cC: \ff^{k(k-1)}\to \ff^{l n}$ encodes $k(k-1)$ symbols of $\ff$ into a codeword of length $n$ with each coordinate formed
of $l$ symbols. To define this mapping, form a matrix $M=[S_1\,|\, S_2]^T$, where $S_1,S_2$ are symmetric matrices of order $l$. The number of unique symbols in $M$ equals $2\binom{l+1}2=k(k-1)$. Next let $x_i,i=1,\dots,n$ be distinct elements of $\ff,$ let
   $$
   \Phi=[\phi_1^T,\dots,\phi_n^T]^T
   $$
be a Vandermonde matrix with rows of the form $\phi_i=(1,x_i,\dots,x_i^{l-1})$ and take $\Lambda=\text{diag}(\lambda_1,\dots,\lambda_n)$ with $\lambda_i=x_i^l,i=1,\dots,n.$ Now form an $n\times 2l$ matrix 
$\Psi=[\Phi, \Lambda \Phi]$ The encoding mapping $\cC$ sends the matrix $M$ to $C=\Psi M$, which is an $n \times l$ matrix, 
and thus the contents of the node $v_i, i=1,\dots,n$ is given by the product 
  \begin{equation}\label{eq:nc}
  C_i:=[\phi_i,x_i^l \phi_i]M=\phi_i S_1+\lambda_i\phi_i S_2.
  \end{equation}

To describe the repair procedure from \cite{Rashmi11}, suppose without loss of generality that the helper nodes form the set $D=\{1,\dots,d\}$ and that the failed node's index is $f\in[n]\backslash[D].$ The original node repair (erasure correction) procedure proposed in \cite{Rashmi11} proceeds as follows. The information downloaded by the failed node $v_f$ from the helper node $i\in D$ is given by $(\phi_iS_1+\lambda_i\phi_iS_2)\phi_f^T,$ i.e., each helper node provides one symbol of $\ff$. Thus, the failed node
downloads a $d$-dimensional vector $y=y_{f,D}$ given by
    \begin{equation}\label{eq:y}
    y=\Psi_D M\phi_f^T=\Psi_D\begin{bmatrix}S_1\phi_f^T\\S_2\phi_f^T\end{bmatrix},
    \end{equation}
where $\Psi_D$ is the submatrix of $\Psi$ formed of the first $d=2l$ rows. The matrix $\Psi_D$ is square $d\times d$ and it is invertible by construction, so we can compute the vectors $(S_1\phi_f^T)^T=\phi_f S_1$ and $(S_2\phi_f^T)^T=\phi_f S_2$. 
By \eqref{eq:nc} the sum $\phi_f S_1+\lambda_f\phi_f S_2$ equals $C_f,$ and
this completes the repair process. 

Now we will modify the repair procedure in a way that supports processing the information received by the nodes 
in the repair tree as it is passed to the failed node $v_f.$ Note that by \eqref{eq:y}
   \begin{equation}\label{eq:-1}
   \phi_f M^T=y^T (\Psi_D^{T})^{-1}.
    \end{equation}
 Using \eqref{eq:nc},\,\eqref{eq:-1}, the contents of the node $v_f$ can be written as
\begin{align*}
   C_f&=\phi_f M^T \begin{bmatrix}
	I_{l} \\ \lambda_f I_{l}\end{bmatrix}=y^T (\Psi_D^{T})^{-1} \begin{bmatrix}
	I_{l} \\ \lambda_f I_{l}\end{bmatrix}.
\end{align*}   
Introduce a $d\times l$ matrix
   $
   U:=(\Psi_D^{T})^{-1} \begin{bmatrix}
	I_{l} \\ \lambda_f I_{l}\end{bmatrix}
   $
and denote its rows by $U_i,i=1,\dots,d,$ then
we have
  \begin{equation}\label{eq:sep}
  C_f=\sum_{i=1}^d y_i U_i.
  \end{equation}
Note that the matrix $U$ does not depend on the codeword, and can be precomputed. Overall this rewriting of the
repair process \eqref{eq:y} enables us to separate the contributions of the helper nodes, and offers savings in the
communication cost of repair. Recalling our notation $D^\ast(v_i),$ suppose that, instead of transmitting the symbol $y_i$ to its parent, the node  transmits the sum $\sum_{j\in D^\ast(v_i)}y_j U_j.$ Since we are now moving vectors rather than individual symbols along the edges of $T_f$,
this may seem wasteful; however remember that the symbols are relayed many times, and that from some point on,
the repair process has to move at least $l$ symbols along the edge by Lemma \ref{lemma:bound}. 
To justify the savings, suppose that $|D^\ast(v_i)|\ge d-k+2=k,$ then forwarding the symbols $(y_j, j\in D^\ast(v_i))$ from $v_i$ to its predecessor in $T_f$
amounts to sending at least $k$ symbols, whereas transmitting the sum $\sum_{j\in D^\ast(v_i)}y_j U_j$ requires
$l=k-1$ transmissions.

Therefore, the communication for repair can be summarized as follows. First, the leaf nodes in $T_f$ send their
symbols $y_i$ one level up, then the nodes that received these symbols send them together with their symbols
$y_i$, etc. If at any stage a node $v_i$ has $d-k+1$ or more descendants, then it switches to transmitting 
\vspace*{-.05in} \begin{equation}\label{eq:lc}
 \sum_{j\in D^\ast(v_i)}y_j U_j.
 \vspace*{-.05in}\end{equation}
  Finally if a node $v_i$ received a vector $\sum_{j\in D(v_i)}y_j U_j$ from 
its immediate descendant, it adds to it the vector $y_i U_i$ and forwards it to its parent in $T_f.$

In summary, we have shown that, for every node $v_i\in T_f$ with $|D(v_i)|\ge d-k+1$ descendants in $T_f$ there exists a repair procedure under which $v_i$ transmits exactly $l$ symbols of $\ff$ to its parent in $T_f.$
This proves the following theorem.
\begin{theorem}\label{thm:PM_repair} Suppose a codeword of a PM code $\cC$ is written on the vertices of a graph $G,$ and let $T_f$ be the
repair tree of a failed node $v_f.$ There exists an explicit repair procedure that achieves the lower bound in (\ref{eq:eqtn25}) with equality. 
\end{theorem}
To match the above procedure to the bound (\ref{eq:eqtn25}), recall that each helper node in the PM code construction provides one symbol of $\ff$ for repair.

%
%

\subsection{Examples of graphs} Let us give a few examples in which the proposed repair procedure gains in communication complexity over the AF repair. For simplicity we will assume that each helper node provides one symbol of $\ff$ for the repair of $v_f$.

\vspace*{.05in}1. Suppose that the repair tree $T_f$ is a {\em star} with $d$ rays in which $v_f$ is one of the leaves and the remaining $d$ vertices serve as the helper nodes. Using the AF repair, each of the nonerased leaves sends its symbol to the center, which then sends $d$ symbols to $v_f,$ so $\beta_{\text{\rm AF}}=2d-1=4k-5.$ At the same time, the repair bandwidth with intermediate processing equals
$\beta_{\text{\rm IP}}=3k-4$ because the symbols of the helpers other than the center are aggregated using \eqref{eq:lc} before relaying to $v_f.$ Another elementary example, which results in a similar improvement, arises when the repair tree $T_f$ is a {\em path} on $d+1$ vertices.

\vspace*{.05in}2. {\em Regular tree}. Suppose that $G$ is an $(r+1)$-regular graph, and the repair tree $T_f$ of every node is $(r+1)$-regular 
as shown in Fig.1. We need to take the depth $t$ of the tree to satisfy 
   $
   (r+1)\sum_{i=0}^{t-1} r^i\ge d;
   $
suppose for simplicity that this holds with equality.   
The communication complexity of the AF repair is \par
 $$
			\beta_{\text{\rm AF}}	= td-(r+1)\sum_{i=0}^{t-2} (t-i-1)r^{i}.
  $$
Suppose that $r>d-k+1$, then from the next to last layer 
we can switch to uploading the linear combination of the form \eqref{eq:lc}, resulting in the repair bandwidth
	$
			\beta_{\text{\rm IP}}
			= d+(d-k)(r+1)\sum_{i=0}^{t-2}r^i.
	$
The difference
	$$	\beta_{\text{\rm AF}}-\beta_{\text{\rm IP}} = (t-1)d - (r+1)\sum_{i=0}^{t-2}((d-k)+(t-i-1))r^i
	$$
is positive if $\frac{d-k}{d}$ is small, i.e., if $d\ge k$ is close to $k$. Note that the regime of small $d-k$  arises
also as a sufficient condition of repair bandwidth savings for random graphs in Sec.~\ref{sec:random}.

\vspace*{.05in}3. {\em Galton-Watson tree.} Having in mind a scenario in which the helper nodes are chosen randomly and independently by the nodes already included in the repair tree $T_f,$ suppose that it is constructed following a branching process with the root $v_f,$ resulting in 
a Galton-Watson ensemble of random trees $\cT_f.$ 
In this example we choose a simple ``offspring pmf'' under which a node has 1 or 2 descendants with probability $p$ and $1-p,$
respectively. Let $Z_i=|\Gamma_i(v_f)|$ be the total number of vertices in layer $i$ of $\cT_f.$
Thus, $\Pr(Z_1=1) = p = 1-\Pr(Z_1=2)$ where $p \in (0,1)$ is chosen to satisfy $m:=\mathbb{E}(Z_1) = 2-p >1$ so that we are operating in the {\em supercritical regime.} Assuming that a tree of depth $t$ suffices for repair, we have 
   \begin{gather*}
\beta_{AF} = td - \sum_{i=1}^{t-1}(t-i)Z_i;\;\;
  \mathbb{E}[\beta_{\text{\rm AF}}] = td-\sum_{i=1}^{t-1}(t-i)m^i.
   \end{gather*}
If we assume that the intermediate processing technique can be applied to layers $i, 1\le i\le s,$ then an easy calculation yields
    $$	
    \mathbb{E}[\beta_{\text{IP}}]
	= (t-s)d + (d-k+1-t+s)\sum_{i=1}^sm^i - \sum_{i=s+1}^{t-1}(t-i)m^i
   $$
and so 
  $$
	\mathbb{E}[\beta_{\text{\rm AF}}-\beta_{\text{\rm IP}}] 	= sd-\sum_{i=1}^s(2-p)^i(d-k+1+s-i)
  $$
which is positive for small values of $d-k$ and large $d$. 

\remove{5. {\em Expander graphs.} Suppose that the graph $G=(V,E), |V|=n$ is $r$-regular and is $(\gamma,\alpha)$-expanding, i.e., for any $S \subset V$ with $|S| \le \gamma n, \gamma \in (0,1/2)$ the neighborhood of $\cN(S) = \{v \in V\backslash S : (u,v) \in E, u \in S\}$ 
satisfies $|\cN(S)| \ge \alpha |S|$ where $\alpha$ is the expansion constant. Consider such a graph and let $d$ be chosen such that a 3-layer repair tree is sufficient. By the expansion property the set $\Gamma_1(v_f)$ (of size $r$) is connected to a set of nodes $\cN(\Gamma_1(v_f))$ of size at least $\alpha r.$ Further, $\Gamma_{2}(v_f)$ satisfies $\alpha(r-1) \le |\Gamma_{2}(v_f)| \le r( r-1)$. So,
\begin{align*}
	\beta_{\text{\rm AF}} &= 3d-\sum_{i=1}^{2}(3-i)|\Gamma_{i}(v_f)| \\
	&\le 3d - \alpha r +1 -2r
\end{align*}
If $d-k \le r$, then intermediate processing can be started from the penultimate layer and
\begin{align*}
	\beta_{\text{\rm IP}} &= \left[d-(|\Gamma_{1}(v_f)|+|\Gamma_{2}(v_f)|)\right]+(d-k+1)\sum_{i=1}^{2}|\Gamma_i(v_f)|\\
	& = d+(d-k)(|\Gamma_{1}(v_f)|+|\Gamma_{2}(v_f)|)\\
	& \le d+(d-k)(r+ r(r-1))\\
	& = d+(d-k)r^2
\end{align*}
The difference between the two upper bounds is  
$$2d-(\alpha+2+(d-k)r)r+1$$ which is positive for large $d$ and small $\alpha,r$ and $d-k$.}

%
%

\subsection{Diagonal-matrix MSR codes}
While the product-matrix codes are limited by the code rate $k/n<1/2$, the construction of \cite{Ye17} removes this limitation, providing
explicit families of exact-repair MSR codes for all possible values of $n-1\ge d \ge k$. 
\remove{Under full connectivity these codes also have the optimal 
repair bandwidth (the lowest communication cost of repair of a single or multiple nodes). 
We briefly recap this construction to argue that it is also amenable to a modification of the repair protocol that makes it suitable for 
repair on graphs.}

The codes in \cite{Ye17} are defined in terms of the parity-check matrix which has a block diagonal structure. 
Below we assume that the parameters of the $(n,k,l)$ array code $\cC$ are fixed, 
and that $d=n-1, l=r^n,$ where $r:=n-k.$ The code is defined over a finite field $\ff$ of size at least $rn.$ 
Let $\{\lambda_{i,j}\}_{i \in [n],j=0,1,\dots,r-1}$ be $rn$ distinct elements of $\ff.$ For an integer $a\in\{0,1,\dots,l-1\}$ 
let $a_i$ be the $i$-th digit of its $r$-ary expansion. For $i=1,2,\dots,n$ define the matrix $A_i=\diag(\lambda_{i,a_i}, a=0,\dots,l-1).$
The code $\cC$ is formed of the codewords $C=(C_1,\dots,C_n)\in (\ff^l)^n$ that satisfy the following set of $r$ parity-check equations:
\begin{equation}\label{eq:eqtn27}
	\sum_{i=1}^n A_i^{t-1}C_i = 0 , \quad t=1,\dots,r.
\end{equation}
Let $C_i=(c_{i,a},a=0,\dots,l-1)^T$. Since the matrices $A_i$ are diagonal, the parity check equations \eqref{eq:eqtn27} take the form
\begin{equation}\label{eq:eqtn29}
	\sum_{i=1}^n\lambda_{i,a_i}^{t-1}c_{i,a} = 0, \quad t = 1,\dots,r,\;a = 0,1,\dots,l-1.
\end{equation}

The node repair with no communication constraints proceeds as follows. Assume that the node $i \in [n]$ has failed. 
We partition the set of coordinates $(c_{i,a})$ into groups of size $r$ whose indices differ only in the $i$th entry.
Namely, start with some $a\in\{0,\dots,l-1\}$ and consider the set of indices $a(i,u) = (a_n,\dots,a_{i+1},u,a_{i-1},\dots,a_1)$, $u=0,1,\dots,r-1.$ The information downloaded from the helper
node $j\in [n]\backslash\{i\}$ is given by
   $
	\mu_{j,i}^{(a)} = \sum_{u=0}^{r-1} c_{j,a(i,u)}.
   $
Writing \eqref{eq:eqtn29} for each of the indices $a(i,u),$ we obtain
   $$
	\lambda_{i,u}^tc_{i,a(i,u)} + \sum_{j\ne i}\lambda_{j,a_j}^tc_{j,a(i,u)}=0, \quad t = 0,1,\dots,r-1.
	$$
Summing these equations on $u$ and writing the result in matrix form, we obtain the relation
\begin{equation}\label{eq:eqtn32}
	\hspace*{-.07in}\left[ \begin{array}{*{4}{@{\hspace*{.03in}}c}}
		1 & 1 & \dots & 1\\
		\lambda_{i,0} & \lambda_{i,1}& \dots & \lambda_{i,r-1}\\
		\vdots & \vdots & \ddots & \vdots\\
		\lambda_{i,0}^{r-1} & \lambda_{i,1}^{r-1}& \dots & \lambda_{i,r-1}^{r-1}
		\end{array}	\right]\!\!
		\begin{bmatrix}
	c_{i,a(i,0)}\\c_{i,a(i,1)}\\\vdots \\c_{i,a(i,r-1)}\end{bmatrix} \!\!=\! -\begin{bmatrix}
		\sum_{j\ne i}\mu_{j,i}^{(a)}\\
		\sum_{j\ne i}\lambda_{j,a_j}\mu_{j,i}^{(a)}\\
		\vdots\\
		\sum_{j\ne i}\lambda_{j,a_j}^{r-1}\mu_{j,i}^{(a)}
	\end{bmatrix}\!.
\end{equation} 
This equation permits recovery of the symbols $c_{i,a(i,u)}, 0\le u\le r-1$ of the failed coordinate, and varying $a$, we 
recover the other groups of coordinates in the same manner.

To adapt this procedure to repair on graphs, assume that the failed node is $i=n$ and write the vector on the right-hand side of \eqref{eq:eqtn32} as 
     $[\mu_{1,n}^{(a)},\mu_{2,n}^{(a)},\dots,\mu_{n-1,n}^{(a)}] V_1^T,$ 
where $$V_1:=\text{Vandermonde}(\lambda_{1,a_1}, \lambda_{2,a_2},  \dots,  \lambda_{n-1,a_{n-1}})$$ is an $r\times(n-1)$ Vandermonde matrix with columns defined by the arguments. The matrix on the left in \eqref{eq:eqtn32} is also Vandermonde, denote it by $V_2.$ With these notations, \eqref{eq:eqtn32} can be rewritten as
    \begin{multline*}
   [c_{n,a(n,0)},c_{n,a(n,1)},\dots, c_{n,a(n,r-1)}]V_2^T\\=-[\mu_{1,n}^{(a)},\mu_{2,n}^{(a)},\dots,\mu_{n-1,n}^{(a)}] V_1^T
   \end{multline*}
or 
 \begin{align}
  [c_{n,a(n,0)},&c_{n,a(n,1)},\dots, c_{n,a(n,r-1)}]\nonumber\\&=[\mu_{1,n}^{(a)},\mu_{2,n}^{(a)},\dots,\mu_{n-1,n}^{(a)}] U \nonumber\\
     &=\sum_{j=1}^{n-1} \mu_{j,n}^{(a)} U_j \label{eq:DM}
 \end{align}
where we denoted $U:=-V_1^T (V_2^T)^{-1}$ and $U_j$ is the $j$th row of $U$.
This representation is essentially the same as \eqref{eq:sep}, and hence 
the generic distributed repair scheme described in Sec.~\ref{sec:MSR} applies to the codes considered in this section. 
Specifically, the matrix $U$ is independent of the codeword, and can be computed in advance, and once a node $v$ in the repair tree
has $d-k+1$ or more descendants, it switches to transmitting $\sum_{j\in D^\ast(v)} \mu_{j,n}^{(a)} U_j$. This procedure supports repair
bandwidth gains over the AF strategy for each of the groups of the node components mentioned above.

\vspace*{.1in}
\subsection{Node repair for general linear array codes}
From the examples in the previous sections it is clear that the graph-based repair procedure defined in \eqref{eq:lc} applies to any 
$\ff$-linear MSR code for which the information downloaded from the helper nodes is an $\ff$-linear function of their contents 
(all the known MSR codes are such). Indeed, the download operation can be written as $C(D)U$, where $C(D)$ is the contents of the helper nodes and $U$ represents the linear transformation of the form \eqref{eq:lc}. Once we reach the helper nodes in $T_f$ with at least $d-k+1$ descendants, then we can switch to relaying linear combinations rather than the contents of the helper nodes. The savings in repair bandwidth will be the same as for the two constructions considered above in this section.

\vspace*{.1in}\emph{Remark (MBR codes)}: For the other extremal point of the storage-bandwidth trade-off \cite{Dimakis10}, i.e., the Minimum Bandwidth Regenerating codes, the AF repair strategy is optimal in terms of the repair bandwidth because the amount of downloaded information is minimized by the code design. 

%
%

\section{Node repair for multiple failures}\label{sec:multiple}

In this section we present a code construction for the repair of multiple nodes that attains the lower bound of 
Lemma~\ref{lemma:bound_mult}. We begin with specifying our assumptions. Suppose that the data is stored on a connected graph $G(V,E)$, and 
$F\subset V$ is a set of failed vertices of size $h$. Further, let $D, |D|=d$ be the subset of helper nodes. The data is encoded using an $(n,k,d,l)$
MSR code, where $n=|V|$ is the number of vertices. The encoding scheme that we present below further assumes that the 
communication from $D$ to $F$ passes through some fixed node $w\in D$ as shown in Fig.~2 for $h=2$ and $F=\{v_1,v_2\}.$ This assumption, taken to fit the 
structure behind Lemma~\ref{lemma:bound_mult}, suggests that we perform simple relaying along the path(s) from $w$ to the failed vertices.
The repair process becomes more complicated if the failed vertices have different access points to $D$, and we do not consider it
here. We further assume that the set $D$ spans a connected subgraph $G_D\subset G$ and denote by $T_w$ a (rooted) spanning tree of $G_D$
with root $w$. Finally, denote by $D_{w}(v)$ the set of descendants of $v\in V(T_w)$ in the tree $T_w$ and let $D^\ast_w(v)=D_w(v)\cup
\{v\}.$

Under these assumptions it is possible to write out a bound on the communication complexity of repair within the set of the helper nodes
until the data reaches the node $w$ (after that the data is no longer processed until it reaches the nodes in $F$). The following
proposition is an obvious extension of the bound \eqref{eq:eqtn25}.
\begin{proposition}\label{prop:Cm} Let $J_{w,h}=\{v\in V(T_w): D_w(v)\ge d-k+h\}.$ The total communication along the edges of $T_w$ for repair
of the nodes in $F$ is bounded below as 
   \begin{equation*}
   \beta(D)\ge |J_{w,h}|l + \sum_{v\in V(T_w)\backslash J_{w,h}}\frac{|D^\ast_w(v)| l}{d-k+h}.
   \end{equation*}
\end{proposition}

Below we present a construction of codes and a repair scheme that meets this bound with equality, attaining the minimum possible 
communication complexity of repair of the nodes in $F$ under the assumptions discussed above (it is possible that removing these
assumptions enables one to further lower the communication cost).
The scheme relies on the idea of {\em cooperative repair} \cite{ShumHu2013}.
In this setting, under the full connectivity assumption, two or more failed nodes connect directly to the same set of helpers, evaluate partial information about their contents, and then exchange the results to complete the repair. 

\vspace*{.1in}
\begin{figure}[th]\begin{center}\scalebox{1.7}
{\begin{tikzpicture}
 [
 v_style/.style={circle, draw,  fill=black, minimum size=0.01cm,scale=0.15},
 v_style1/.style={circle, draw,  fill=none, minimum size=0.01cm,scale=0.15},
 v_style2/.style={circle, draw=none,   fill=none}
 ]
 
  \node[circle, draw,red, fill=red, minimum size=0.01cm,scale=0.15,label=left:{ {\scriptsize$v_1$}}] (v1) at (-1,1.0) {};
   \node[circle, draw,red, fill=red, minimum size=0.01cm,scale=0.15, label=right:{\scriptsize $v_2$}] (v2) at (1,1.0) {};
    \node[v_style] (v3) at (0,0) {};
 \node[circle, draw,blue, fill=blue, minimum size=0.01cm,scale=0.15,label={[xshift=-0.2cm, yshift=-.1cm]{\scriptsize $\;\;w$}}] (w) at (0,-0.98) {};
  \path [draw=black, snake it]
     (-1,1) -- (0,0);
  \path [draw=black, snake it]
     (1,1) -- (0,0);
   \path [draw=black, snake it]
     (0,0) -- (0,-0.98);
         
  \node[ellipse, draw, minimum width = 1in, 
    minimum height = .5in,label={[xshift=.0cm, yshift=-1.6cm]{\tiny\sc Helper set $D$}}] (e) at (0,-1.62) {};
 
  \node[v_style1] (h1) at (-0.6,-1.5){};
  \node[v_style1] (h2) at (0,-1.5){};
  \node[v_style1] (h3) at (0.6,-1.5){};
  \path (w) edge [black,ultra thin] (h1);
  \path (w) edge [black,ultra thin] (h2);
  \path (w) edge [black,ultra thin] (h3);
  
  \node[v_style2] (h4) at (-1.1,-2){};
  \node[v_style2] (h5) at (-0.5,-2.05){};
  \node[v_style2] (h6) at (-.2,-2){};
  \node[v_style2] (h7) at (-.15,-2){};
  \node[v_style2] (h8) at (.2,-2){};
  \node[v_style2] (h9) at (0.3,-2){};
  \node[v_style2] (h10) at (0.9,-2){};

  \path (h1) edge [black,ultra thin] (h4);
  \path (h1) edge [black,ultra thin] (h5);
  \path (h1) edge [black,ultra thin] (h6);
  \path (h2) edge [black,ultra thin] (h7);
  \path (h2) edge [black,ultra thin] (h8);
  \path (h3) edge [black,ultra thin] (h9);
  \path (h3) edge [black,ultra thin] (h10);
  
  \node[circle, draw,  fill=black, minimum size=0.01cm,scale=0.05] at (-.4,-2) {};
  \node[circle, draw,  fill=black, minimum size=0.01cm,scale=0.05] at (-.2,-2) {};
  \node[circle, draw,  fill=black, minimum size=0.01cm,scale=0.05] at (-.0,-2) {};
  \node[circle, draw,  fill=black, minimum size=0.01cm,scale=0.05] at (.2,-2) {};
  \node[circle, draw,  fill=black, minimum size=0.01cm,scale=0.05] at (.4,-2) {};
  
\end{tikzpicture}}
\end{center}
\centerline{\footnotesize {\sc Fig.2:} Graph topology for repair of multiple nodes}
\end{figure}

\vspace*{.1in}

We use this idea for repair on graphs wherein the information from helpers is transmitted along some path to the failed node, relying
on a family of cooperative codes constructed recently in \cite{Ye20}.
The savings come from the fact that in the course of this transmission we can perform intermediate processing rather than simple relaying.
In our presentation
for simplicity we assume that $d=k+1, h=2$ as in Fig. 2. At the same time it will be obvious that the technique applies to all other 
feasible parameter regimes.

Let $\cC$ be the $[n,k,d=k+1,l=3\times 2^n]$ {\em cooperative repair} MSR code from the family constructed in \cite{Ye20}\footnote{We could use
other code families, for instance, the codes from \cite{Ye19}.}. 
Every coordinate $C_i$ of a codeword $C = (C_1,\dots,C_n)\in (\ff^l)^n$ is a vector $\{c_{i,b,a}: b \in \{1,2,3\}, a \in \{0,1,\dots 2^n-1\}\}$. 
For $2n$ distinct field elements $\{\lambda_{i,j}: i \in [n], j \in \{0,1\}\}$, the parity check equations that define the code are
   \begin{multline}\label{eq:eqtn_coop_pc}
    \sum_{i=1}^n \lambda_{i,a_i}^tc_{i,b,a} = 0 \;\;\;\forall\;\;\;t \in \{0,1,\dots, n-k-1\},\\\;\;a \in \{0,1,\dots,2^n-1\},\;\; b\in                    
           \{1,2,3\},
   \end{multline}
where $a_i$ is the $i$-th coordinate in the binary representation of $a$. Below we use the notation $a(i,a_i\oplus 1)$ to denote the number obtained from $a$ by flipping the $i$th bit in its binary expansion.
Assume that the failed nodes correspond to coordinates 1 and 2 and fix a value of $a \in \{0,1,\dots,2^n-1\}$.
The standard cooperative repair under direct connectivity (on a complete graph) proceeds in two steps. 
In step 1, helper node $i$ sends $\{c_{i,1,a}+c_{i,2,a(1,a_1\oplus 1)}: a \in \{0,1,\dots,2^n-1\}\}$ to node 1 and $\{c_{i,1,a}
+c_{i,3,a(2,a_2 \oplus 1)}: a \in \{0,1,\dots,2^n-1\}\}$ to node 2. Using $a$ with $b=1$ and $a(1,a_1\oplus 1)$ with $b=2$ in \eqref{eq:eqtn_coop_pc} and summing the corresponding equations, we obtain
\begin{align}
\lambda_{1,a_1}^tc_{1,1,a}&+\lambda_{1,a_1\oplus 1}^tc_{1,2,a(1,a_1\oplus 1)} \nonumber\\
&+ \lambda_{2,a_2}^t(c_{2,1,a}+c_{2,2,a(1,a_1\oplus 1)}) 
   \nonumber\\ 
&+ \sum_{i=3}^n \lambda_{i,a_i}^t(c_{i,1,a}+c_{i,2,a(1,a_1\oplus 1)})=0 \label{eq:crs}
\end{align}
for all $ t \in \{0,1,\dots,n-k-1\}.$
Equations \eqref{eq:crs} form a set of parity checks of an $(n+1,k+1)$ Reed-Solomon code,
and hence knowing $c_{i,1,a}+c_{i,2,a(1,a_1\oplus 1)}$ at $k+1$ positions allows node 1 to recover $c_{1,1,a}, c_{1,2,a(1,a_1\oplus 1)}$ and $(c_{2,1,a}+c_{2,2,a(1,a_1\oplus 1)})$. A similar argument shows that node 2 can recover $c_{2,1,a}, c_{2,3,a(2,a_2\oplus 1)}$ and $(c_{1,1,a}+c_{1,3,a(2,a_2\oplus 1)})$. In step 2 of the repair, node 1 sends $(c_{2,1,a}+c_{2,2,a(1,a_1\oplus 1)})$ to node 2 and node 2 sends $(c_{1,1,a}+c_{1,3,a(2,a_2\oplus 1)})$ to node 1, which completes the repair of both node 1 and 2; for details see \cite{Ye20}.

To see how intermediate processing at the nodes of the tree $T_w$ can simplify repair on a graph of the type shown in Fig.~2, 
observe that the first step above, node 1 seeks to learn three code symbols (namely $c_{1,1,a}, c_{1,2,a(1,a_1\oplus 1)}$ and $(c_{2,1,a}+c_{2,2,a(1,a_1\oplus 1)})$) of the $(n+1,k+1)$ RS codeword, and it does so by collecting $k+1$ symbols from $k+1$ helper nodes. In an RS code, once we know any $k+1$ coordinates, all the other coordinates of the codeword can be computed via Lagrange interpolation and subsequent evaluation. This can be expressed in matrix form as follows:
    \begin{multline*}
\begin{bmatrix}
	c_{1,1,a}\\ c_{1,2,a(1,a_1\oplus 1)}\\(c_{2,1,a}+c_{2,2,a(1,a_1\oplus 1)})
\end{bmatrix} \\= 
[U_1 \; U_2 \;\dots \; U_{k+1}]
\begin{bmatrix}
(c_{i_1,1,a}+c_{i_1,2,a(1,a_1\oplus 1)})\\
(c_{i_2,1,a}+c_{i_2,2,a(1,a_1\oplus 1)})\\
\vdots \\
(c_{i_{k+1},1,a}+c_{i{k+1},2,a(1,a_1\oplus 1)})
\end{bmatrix},
    \end{multline*}
where $i_1,i_2,\dots,i_{k+1}$ are the helper nodes and the matrix $U = [U_1 \; U_2 \;\dots \; U_{k+1}]$ 
is obtained by multiplying an inverse Vandermonde matrix (Lagrange interpolation) and a matrix corresponding to evaluating the obtained
polynomial at the three coordinates being sought.
Since the matrix $U$ can again be pre-computed, a node that has collected the values 
$(c_{i,1,a}+c_{i,2,a(1,a_1\oplus 1)})$ from three or more 
helper nodes, can start transmitting the corresponding linear combinations, much in the same way as was done in Section \ref{sec:MSR}. The above procedure is repeated for node 2 with appropriate adjustments to the subscripts in the last displayed equation. Step 2 of the repair process is unchanged from that of the standard cooperative repair, and it yields no communication savings. Exactly as in the case of a single failed node, viz., Theorem \ref{thm:PM_repair}, we can show that this procedure meets the bound of Lemma \ref{lemma:bound_mult}.

%
%

\section{Repair with information exchange among the helpers}\label{sec:intra}
The bounds and constructions presented earlier in this paper are focused on communication from the helper nodes to the
failed node. In this section we consider a more general problem (and potential savings in the repair cost) when the helper nodes may communicate with each other before transmitting the information to the failed node. Recall that a variant of this problem was considered earlier in the literature under very specific assumptions: The nodes in the storage cluster are organized in subsets, called {\em racks}, and communication between the nodes in the rack does not count toward the repair bandwidth. This model enables one to derive tighter bounds on the cost of node repair \cite{Hou2018}, and there are families of codes that 
attain these bounds \cite{Chen2020}. 

Another version of information exchange in the context of erasure recovery appeared earlier in the problem of {\em cooperative repair},
already mentioned in the previous section.
In this setting (assuming full connectivity) several failed nodes contact the same set of helpers and process the received information,
gaining some knowledge about their contents and about the contents of the other failed nodes. They then exchange information to complete the repair. This problem, introduced in \cite{ShumHu2013}, is vaguely reminiscent of repair on graphs because different nodes of the
encoding acquire partial information about the contents of other nodes. Below we make this link more precise by presenting an
example of node repair on graphs motivated by cooperative repair (albeit in a rather restricted setting).

We begin with establishing a framework for finding a lower bound on the total communication complexity of repair for general
graphs. Let us define some additional notation. Let an $(n,k,d,l)$ MSR code be defined on a connected graph $G=(V,E)$. Assume that the subgraph $G_{f,D}= (V_{f,D},E_{f,D})$ 
spanned by the failed node $v_f$ and the set of helper nodes $D$ is connected. 
Construct a directed graph $\bar{G}_{f,D}= ({V}_{f,D},\bar{E}_{f,D})$ as follows:
\begin{itemize}
	\item For every edge $(u,v) \in E_{f,D}$ with $u,v \in D$, add the two directed edges $(u,v)$ and $(v,u)$ to $\bar{E}_{f,D}$.
	\item For every edge $(u,v_f) \in E_{f,D}$, add the directed edge $(u,v_f)$ to $\bar{E}_{f,D}$.
\end{itemize}

For an arbitrary communication protocol 
that repairs the failed node $v_f$, let $X_{u,v}$, for $(u,v) \in \bar{E}_{f,D}$, be the total number of symbols 
sent along the edge $(u,v)$ during the complete protocol. Fix an order of the edges in $\bar E=\bar E_{f,D}$ and 
let $\bar{X}$ be the vector of $X_{u,v}$'s. Let $\cP^\ast(D)$ be the set of all non-empty subsets of $D$. Define a binary matrix $M$ of size 
$(2^d-1) \times |\bar{E}|$ by setting $M_{S,(u,v)}={\mathbbm 1}(u\in S\wedge v\in S^c),$ where $S^c={V}_{f,D}\backslash S.$ The rows of $M$
are characteristic vectors of the cuts $(S,S^c)$.
Let $\bar b\in\reals^{|2^d-1|}$ with $\bar b_S=\beta\cdot\min\{d-k+1,|S|\}$ for all $S\in\cP^\ast(D).$

\begin{proposition}\label{prop:LP}
For the failed node $v_f$ and helper nodes $D$, the total communication complexity of repair is bounded below by the solution to the following linear program with $|\bar E|$ variables and $2^d-1$ constraints:
\begin{equation*}
	\begin{array}{ll@{}ll}
		\text{minimize}  & \displaystyle \bm{1}^T \bar{X} &\\
		\text{subject to}& \displaystyle M\bar{X} \ge \bar{b},\\
		&                                               \bar{X} \ge 0.
	\end{array}	
\end{equation*}
\end{proposition}
\begin{proof} We only need to justify the inequality $M\bar{X} \ge \bar{b}.$
For any set $S \in \cP^\ast(D)$, Lemma \ref{lemma:bound} implies
\begin{equation*}
	\sum_{\begin{substack}{(u,v)\in \bar{E}_{f,D}\\ u \in S,v \in S^c}\end{substack}} X_{u,v} \ge R_S^f \ge \min\{d-k+1,|S|\}\beta.
\end{equation*}
Collecting these inequalities for all $S \in \cP^\ast(D)$, we obtain the claimed set of constraints.
\end{proof}
The key observation here is that the quantity $R_A^f$ in Lemma \ref{lemma:bound} represents the total outflow of 
information transmitted from the set of nodes $A$ for the repair, and hence the bounds still hold irrespective of the 
communication among the nodes in set $A$.

In the remainder of this section we consider two settings in which the bound of this proposition enables one to prove optimality
of communication for recovery while allowing communication between the helper nodes, namely when the failed node has the largest 
and the smallest possible number of helpers, respectively, as immediate neighbors. In both cases we allow arbitrary communication among the helper set.

\subsection{The case of the complete graph} 
This case corresponds to the original repair problem of \cite{Dimakis10}, and the cut-set bound provides
the minimum required download per helper node for the repair of a failed node. 
In this model, the transmitted data of each helper node is a function of its own stored content only. 
Can communication complexity be reduced if the helper nodes are allowed to exchange information before communicating with the failed node? 
An easy corollary of Proposition \ref{prop:LP} and Lemma \ref{lemma:bound} implies that in case of MSR codes the answer is negative.

\begin{proposition}
	For the complete graph $K_n$, the communication complexity is bounded below by $d\beta$ and is achieved by having all the helper nodes directly transmit $\beta$ symbols to the failed node.
\end{proposition}
\begin{proof}
Consider the assignment of variables of the LP problem $X^*$ with $X^*_{u,v} = \beta \mathbbm{1}(v=v_f)$.
It is clearly feasible because it corresponds to all the helper nodes transmitting $\beta$ 
symbols to the failed node. Indeed, this assignment satisfies the bounds of Lemma \ref{lemma:bound} and thus also 
the inequality constraints of Proposition~\ref{prop:LP}. Next we show that $X^*$ is optimal by considering the dual LP problem,
which has the form
\begin{equation*}
	\begin{array}{ll@{}ll}
		\text{maximize}  & \displaystyle \bar{b}^T\bar{Y}  &\\
		\text{subject to}& \displaystyle M^T\bar{Y} \le \bm{1},\\
		&                                               \bar{Y} \ge 0\;\;.
	\end{array}	
\end{equation*}
Take the assignment of variables $Y^\ast$ with $Y^*_S=\mathbbm{1}(|S|=1)$ for all $S\subset D.$
Since for two different $S_1 = \{v_1\}$ and $S_2 = \{v_2\}$ any edge $(u,v)\in \bar E$ can belong to at most one of the cuts 
$(S_1,S_1^c)$ or $(S_2,S_2^c)$, we have that $M^T Y^* \le \bm{1}$. Since $\bm{1}^T X^* = \bar{b}^T Y^* = d\beta,$ we conclude that $X^\ast$ is indeed optimal.
\end{proof}

\subsection{The case of two neighbors}\label{sec:coop}
Assume that the information is encoded with an $[n,k,d=k+1,l]$ MSR code and stored on a graph $G(V,E)$ with $|V|=n.$ 
Consider the repair graph (no longer a tree) shown in Fig.3 with the failed node $v_f$ connected to two helper nodes which connect to the remaining subset of the helper set.

\begin{figure}[th]\begin{center}\scalebox{0.25}{\begin{tikzpicture}
				[
				vertex_style/.style={circle, draw, fill=red,minimum size=0.25cm,scale=2.0},
				vertex_style1/.style={circle, draw, fill=blue,minimum size=0.25cm,scale=2.0},
				dots/.style={circle, draw,fill=blue, minimum size=0.02cm,scale=0.75}
				]
				
				\useasboundingbox (-10,-12) rectangle (10,2);
				
				\begin{scope}
					
				\node[vertex_style] (v1) at (0,0) {};
				\node[vertex_style1] (v2) at (-3,-5) {};
				\node[vertex_style1] (v3) at (3,-5) {};
				\node[vertex_style1] (v4) at (-10,-10) {};
				\node[vertex_style1] (v5) at (-5,-10) {};
				\node[vertex_style1] (v6) at (5,-10) {};
				\node[vertex_style1] (v7) at (10,-10) {};
				\node[dots] (v8) at (1,-10) {}; 
				\node[dots] (v8) at (0,-10) {}; 
				\node[dots] (v8) at (-1,-10) {}; 
				
				\draw[<-, line width=1mm] (v1) -- (v2);
				\draw[<-, line width=1mm] (v1) -- (v3);
				\draw[<->, line width=1mm] (v2) -- (v3);
				\draw[<->, line width=1mm] (v2) -- (v4);
				\draw[<->, line width=1mm] (v2) -- (v5);
				\draw[<->, line width=1mm] (v2) -- (v6);
				\draw[<->, line width=1mm] (v2) -- (v7);
				\draw[<->, line width=1mm] (v3) -- (v4);
				\draw[<->, line width=1mm] (v3) -- (v5);
				\draw[<->, line width=1mm] (v3) -- (v6);
				\draw[<->, line width=1mm] (v3) -- (v7);
				
				\node[minimum size=0.5cm,scale=3] at (1,1) {$v_f$};
				\node[minimum size=0.5cm,scale=3] at (4.2,-4.5) {$v_2$};
				\node[minimum size=0.5cm,scale=3] at (-4,-4.5) {$v_1$};
				\node[minimum size=0.5cm,scale=3] at (12,-9) {$\Gamma_2(v_f)$};
				
				\draw [decorate,decoration={brace,amplitude=10pt,mirror},line width=1mm]
				(-11,-11) -- (11,-11) node [black,midway,xshift=0, yshift=-1cm,scale=3] 
				{$k-1$ nodes};
				
				\end{scope}
		\end{tikzpicture}}
		\centerline{{\footnotesize {\sc Fig.3:} Repair graph of the node $v_f$ that attains the LP lower bound}}
		\end{center}
\label{fig:2neighbors}
\end{figure}

We will prove that for this graph the minimum required communication for repair equals $(d+1)\beta = (k+2)\beta.$
To show this, assume that the failed node $v_f$ relies on a set $D$ of $k+1$ helpers for repair, and that it is connected to two of them,
denoted $v_1$ and $v_2.$ Assume further that the $k+1$ helpers span a complete graph $K_{k+1},$ where $k$ is the dimension of the
MSR code used for the encoding of the data.
To link this graph to the LP problem of Prop.~\ref{prop:LP}, construct a directed graph by replacing every edge between a pair of helpers with
a pair of opposing directed edges, and make a directed edge from each of $v_1,v_2$ to $v_f$. Thus, the new set of {\em directed} edges is
   $$
   \bar E=\{((v_i,v_j), v_i,v_j\in D), \text{ and }({v_1,v_f}),({v_2,v_f})\}.
   $$
To construct a primal LP program, assign 
    $$
    X^\ast_{(u,v)}=\begin{cases}
         \beta &\text{if }u\in \Gamma_2(v_f)\cup \{v_2\}, v={v_1}\\
         2\beta &\text{if }u=v_1, v=v_f\\
         0 &\text{otherwise}.
         \end{cases}
   $$
This assignment defines a valid repair protocol, so it's a feasible solution of the LP problem which gives the value of the objective function to be
	\begin{equation}
		\bm{1}^TX^* = (d-1)\beta+2\beta = (d+1)\beta
	\end{equation} 
Construct a dual program $Y^*=(Y^*_S)_S$ by setting
    $$
    Y^\ast_S=\begin{cases} \frac1{d-2} &\text{if } |S|=2, S\ne\{v_1,v_2\}\\
      0&\text{otherwise}.
     \end{cases}
   $$   
The vector $Y^\ast$ is a feasible assignment of the dual program. To show this, consider an edge $(u,v)$ with $u\in D.$ Our
argument depends on whether $v\in D$ or $v=v_f.$ In the first case, the row of $M^T$ contains $d-2$ ones which correspond to the
cuts in $G_{f,D}$ that contain the edge $(u,v)$ (there are exactly $d-2$ such cuts), so this row times $Y^\ast$ equals one.
If $v=v_f,$ then $u$ is either $v_1$ or $v_2.$ Say it is $v_1,$ then the row $(v_1,v_f)$ contains $d-1$ ones which correspond to the cuts that contain the edge $(v_1,v_f)$. Further, $Y^\ast=0$ in the coordinate $S=\{v_1,v_2\}$, so the nonzeros in the vector $Y^\ast$ and the $(M^T)_{v_1,v_f}$ overlap in $d-2$ places, again satisfying the constraints of the dual program.

 To compute the value
of the dual problem, note that $Y^\ast\ne 0$ in $\binom{d}{2} - 1 = \frac{(d-2)(d+1)}{2}$ coordinates, and the corresponding entries in $\bar{b}$ are set to $2\beta.$ Thus, $\bar{b}^TY^* = 
\frac{(d-2)(d+1)}{2}\cdot \frac{1}{d-2}\cdot 2\beta = (d+1)\beta,$ which equals the value of the primal problem, proving that the repair protocol
defined by it yields the minimum possible communication complexity. Finally, we argue that if there does not exist a repair protocol that performs better in terms of complexity when the helper nodes form the complete graph, then there cannot exist a repair protocol that performs better for any sub-graph of the complete graph.

The repair bandwidth $(d+1)\beta$ can be attained by sending the data from all the helper nodes but $v_1$ to the node $v_1,$
combining them and passing the result to $v_f$ (which is the IP repair discussed earlier). This repair protocol does not involve 
two-way communication between the neighboring helper nodes. In the appendix we construct another protocol that does involve it, while still having the same communication complexity of repair.

\subsection{Can the repair bandwidth be lower than the IP protocol?}
So far we have not identified cases in which communication among the helper nodes reduces the complexity of repair compared
to the IP protocol. That this may be possible is demonstrated in the next numerical example in which the value of the linear program is below the repair bandwidth of the IP scheme. Note that we still stop short of constructing an actual node repair scheme that would have this value of the communication complexity.

Consider the graph $G_{f,D}$ in Fig.~4 with three direct neighbors of the failed node, and let $d=6, k=5$. We 
assume that $v_f=1,$ and it is directly connected to helper nodes 2, 3 and 4. The six helper nodes form a complete graph $K_6$.
 The IP technique can achieve the complexity of 7 units by transmitting along the spanning tree shown in Fig.~5 and performing IP 
(combining the data) at node 3.

\begin{figure}[th]
\begin{subfigure}{.2\textwidth}
\centering
\scalebox{0.27}{\begin{tikzpicture}
				[
				vertex_style/.style={circle, draw, fill=red,minimum size=0.25cm,scale=2.0},
				vertex_style1/.style={circle, draw, fill=blue,minimum size=0.25cm,scale=2.0},
				dots/.style={circle, draw,fill=blue, minimum size=0.02cm,scale=0.75}
				]
				
				\useasboundingbox (-10,-12) rectangle (10,2);
				
				\begin{scope}
					
					\node[vertex_style] (v1) at (0,0) {};
					\node[vertex_style1] (v2) at (-5,-5) {};
					\node[vertex_style1] (v3) at (0,-5) {};
					\node[vertex_style1] (v4) at (5,-5) {};
					\node[vertex_style1] (v5) at (-5,-10) {};
					\node[vertex_style1] (v6) at (0,-10) {};
					\node[vertex_style1] (v7) at (5,-10) {};
					
					\draw[<-, line width=1mm] (v1) -- (v2);
					\draw[<-, line width=1mm] (v1) -- (v3);
					\draw[<-, line width=1mm] (v1) -- (v4);
					\draw[<->, line width=1mm] (v2) -- (v3);
					\draw[<->, line width=1mm] (v2) to [out=150,in=30]  (v4);
					\draw[<->, line width=1mm] (v2) -- (v5);
					\draw[<->, line width=1mm] (v2) -- (v6);
					\draw[<->, line width=1mm] (v2) -- (v7);
					\draw[<->, line width=1mm] (v3) -- (v4);
					\draw[<->, line width=1mm] (v3) -- (v5);
					\draw[<->, line width=1mm] (v3) -- (v6);
					\draw[<->, line width=1mm] (v3) -- (v7);
					\draw[<->, line width=1mm] (v4) -- (v5);
					\draw[<->, line width=1mm] (v4) -- (v6);
					\draw[<->, line width=1mm] (v4) -- (v7);
					\draw[<->, line width=1mm] (v5) -- (v6);
					\draw[<->, line width=1mm] (v5) to [out=-150,in=-30]  (v7);
					\draw[<->, line width=1mm] (v6) -- (v7);
					\node[minimum size=0.5cm,scale=3] at (1,1) {$1$};
					\node[minimum size=0.5cm,scale=3] at (-6.5,-4.9) {$2$};
					\node[minimum size=0.5cm,scale=3] at (-1,-4.5) {$3$};
					\node[minimum size=0.5cm,scale=3] at (6.5,-4.9) {$4$};
					\node[minimum size=0.5cm,scale=3] at (-6.5,-10.1) {$5$};
					\node[minimum size=0.5cm,scale=3] at (-0.6,-11) {$6$};
					\node[minimum size=0.5cm,scale=3] at (6.5,-10.1) {$7$};
					
				\end{scope}
		\end{tikzpicture}}
\caption*{\hspace*{.3in}\footnotesize {\sc Fig.4:} The repair graph}
\end{subfigure}
\hspace*{.2in}\begin{subfigure}{.2\textwidth}
\centering\scalebox{0.27}{\begin{tikzpicture}
				[
				vertex_style/.style={circle, draw, fill=red,minimum size=0.25cm,scale=2.0},
				vertex_style1/.style={circle, draw, fill=blue,minimum size=0.25cm,scale=2.0},
				dots/.style={circle, draw,fill=blue, minimum size=0.02cm,scale=0.75}
				]
				
				\useasboundingbox (-10,-12) rectangle (10,2);
				
				\begin{scope}
					
					\node[vertex_style] (v1) at (0,0) {};
					\node[vertex_style1] (v2) at (-5,-5) {};
					\node[vertex_style1] (v3) at (0,-5) {};
					\node[vertex_style1] (v4) at (5,-5) {};
					\node[vertex_style1] (v5) at (-5,-10) {};
					\node[vertex_style1] (v6) at (0,-10) {};
					\node[vertex_style1] (v7) at (5,-10) {};

					\draw[<-, line width=1mm, dashed] (v1) -- (v2);
					\draw[<-, line width=1mm] (v1) -- (v3);
					\draw[<-, line width=1mm, dashed] (v1) -- (v4);
					\draw[<->, line width=1mm] (v2) -- (v3);
					\draw[<->, line width=1mm, dashed] (v2) to [out=150,in=30]  (v4);
					\draw[<->, line width=1mm, dashed] (v2) -- (v5);
					\draw[<->, line width=1mm, dashed] (v2) -- (v6);
					\draw[<->, line width=1mm, dashed] (v2) -- (v7);
					\draw[<->, line width=1mm] (v3) -- (v4);
					\draw[<->, line width=1mm] (v3) -- (v5);
					\draw[<->, line width=1mm] (v3) -- (v6);
					\draw[<->, line width=1mm] (v3) -- (v7);
					\draw[<->, line width=1mm, dashed] (v4) -- (v5);
					\draw[<->, line width=1mm, dashed] (v4) -- (v6);
					\draw[<->, line width=1mm, dashed] (v4) -- (v7);
					\draw[<->, line width=1mm, dashed] (v5) -- (v6);
					\draw[<->, line width=1mm, dashed] (v5) to [out=-150,in=-30]  (v7);
					\draw[<->, line width=1mm, dashed] (v6) -- (v7);
					
					\node[minimum size=0.5cm,scale=3] at (1,1) {$1$};
					\node[minimum size=0.5cm,scale=3] at (-6,-5) {$2$};
					\node[minimum size=0.5cm,scale=3] at (-1,-4.5) {$3$};
					\node[minimum size=0.5cm,scale=3] at (6,-5) {$4$};
					\node[minimum size=0.5cm,scale=3] at (-6,-10) {$5$};
					\node[minimum size=0.5cm,scale=3] at (-0.6,-11) {$6$};
					\node[minimum size=0.5cm,scale=3] at (6,-10) {$7$};
				\end{scope}
		\end{tikzpicture}}
\caption*{\hspace*{.3in}\footnotesize {\sc Fig.5:} Repair using IP}
\end{subfigure}
\end{figure}

\vspace*{.2in}
To define the LP problem we construct a directed graph $\bar G_{f,D}$ as explained in the beginning of this section. 
The linear program of Prop.~\ref{prop:LP} in this case has value $6.75$ and the assignments of variables are: the primal program
\begin{equation*}
	X^*_{(u,v)} = \begin{cases}
		1 & \text{if $u \in \{2,3,4\}, v = 1$},\\
		0.5 & \text{if $u,v \in \{5,6,7\}, u \ne v$},\\
		0.25 & \text{if $u \in \{5,6,7\}, v \in \{2,3,4\}$},\\
		0 & \text{otherwise;}
	\end{cases}
\end{equation*}
the dual program:
\begin{equation*}
	Y^*_S = \begin{cases}
		0.125 & \text{if $|S|=2, S \subset \{2,3,4\}$},\\
		0.25 & \text{if $|S|=2, S \not\subset \{2,3,4\}$},\\
		0 & \text{otherwise}
	\end{cases}
\end{equation*}
Many more similar examples can be constructed for small-size graphs.

%
%
 
\section{Node repair on random graphs}\label{sec:random}
In this section we analyze the distributed repair procedure in the case when the underlying graph $G(V,E)$ is sampled from the $\cG_{n,p}$ ensemble, where $0<p<1.$ We denote such a random element from the ensemble as $\bG_{n,p}$. As before, we assume that the coordinates $C_1,\dots,C_n$ of a codeword of an $(n,k,d)$ MSR code are placed on the vertices 
$v_1,\dots,v_n$. The main question that we address is finding relations between the parameters $p,n,k,d$ such that 
graph-based repair of the failed node with high probability results in lower repair bandwidth than the AF strategy. Throughout this section we
assume that each helper node provides one field symbol for the repair of the (single) failed node.

We will assume that $p \gg \frac{\log n}{n}$  because if $\bG_{n,p}$ is not connected, then with positive
probability the node $v_f$ is isolated, and repair is not possible (the notation $f(n)\gg g(n)$ means that $g(n)=o(f(n))$). Furthermore, $\PP_{\cG_{n,p}}(\deg(v_f)\ge d)=\sum_{i=d}^n\binom{n}{i}p^i(1-p)^{n-i},$ which goes to zero for $n\to\infty$ if $d\gg np.$ Thus, overall this is the
parameter regime that may make the graph-based repair (possible and) advantageous over the agnostic AF repair procedure.

Throughout we will assume that $k$ and $d\in \Theta(n),$ and that $\chi(n):=d-k$ is $o(n)$, i.e., $d$ is close to $k$.
For simplicity (without loss of generality) we also assume that each helper node
provides only one symbol of $\ff$ for the repair of the failed node.

We will use the following two results regarding the random Erd{\"o}s-R{\'e}nyi graphs (below $\PP=\PP_{\cG_{n,p}}$).
\begin{lemma}[\!\!\cite{Bollobas81}, p.~50; \cite{FK2016}, Sec.7.1]\label{lemma:diam}
		
		(i) If $p^2n-2\log n \rightarrow \infty,$ and $	n^2(1-p) \rightarrow \infty,$
then $$\PP(\diam(\bG_{n,p})=2)\to 1.$$

(ii) Suppose that the functions $t = t(n) \ge 3$ and $0<p=p(n) <1$ satisfy 
    \begin{gather*}
    (\log n)/t - 3\log \log n \rightarrow \infty, \quad p^t n^{t-1} - 2\log n \rightarrow \infty,\\
			p^{t-1}n^{t-2} - 2\log n \rightarrow -\infty,
	\end{gather*}
then  $\PP(\diam(\bG_{n,p})=t)\to 1.$
\end{lemma}
 \begin{lemma}[\!\!\cite{Chung01}, Lemma 3]\label{lemma:Ni}
Suppose that $p \ge \frac{\log n}{n}$. For any $\epsilon >0$ and all $i=1,\dots,\lfloor\log n\rfloor$
 \begin{gather}
    \PP(|\Gamma_i(x)| \le (1+\epsilon)(np)^i)\ge 1-{1}/{\log^2 n} \label{eq:N>}\\
    \PP(|N_i(x)| \le (1+2\epsilon)(np)^i)\ge 1-{1}/{\log^2 n} . \label{eq:N<}
	\end{gather}	
 \end{lemma}

\subsection{ Repair threshold}\label{sec:threshold}
Let $t$ be a fixed integer. We say that {\em $t$-layer repair} of the failed node $v$ {\em is possible} if 
$$
   \PP(|N_t(v)|\ge d)\to 1 \text{ as $n\to\infty$}
   $$
and call the minimum $t$ for which this holds the {\em threshold depth} for repair. Note that such a $t$ is a function of $n$ and $p$.
   The next proposition establishes a threshold for $t$-layer repair in terms of $p.$ 
   \begin{proposition}\label{lemma:threshold}
   If
   \begin{gather}\label{eq:conditions}
    (np)^{t-1}=o(n), \quad
p^t n^{t-1}-2\log n \rightarrow \infty,
    \end{gather}
    then $t$ is the threshold depth for repair.
   \end{proposition}
\begin{proof} 
To show that $t$-layer repair is possible, we observe that from Lemma~\ref{lemma:diam},
$\PP(\diam(\bG_{n,p}) = t) \rightarrow 1$ for all $t\ge 2$. This implies that for any failed node $v$, all the other nodes in the graph are reachable in at most $t$ 
steps, and in particular, $|N_t(v)|=n>d.$
To show that $t$ is the smallest radius that supports repair, observe that by \eqref{eq:N<}
for any $\epsilon >0$ 
   \begin{equation}\label{eq:->1}
\PP(|N_{t-1}(v)| \le (1+2\epsilon)(np)^{t-1})\to 1.
   \end{equation}
 Since $d$ is a linear function of $n$, we have the inclusion
     $$
     \{|N_{t-1}(v)| \ge d\} \subset \{|N_{t-1}(v)|/n > 0\} \quad(n\to\infty).
     $$
Together with \eqref{eq:->1} this implies that $\PP(|N_{t-1}(v)|/n \ge \gamma) \rightarrow 0$ for any $\gamma>0.$
\end{proof}
{\em Remark:} Given $t$, the conditions \eqref{eq:conditions} are satisfied if
  \begin{equation}\label{eq:p(n)}
  n^{-(t-1)/t}g(n) \ll p(n) \ll n^{-(t-2)/(t-1)},
  \end{equation}
where $g(n)\gg(2\log n)^{1/t}.$ Rephrasing Prop.~\ref{lemma:threshold}, we could say that for a given repair depth $t$ the probability $p(n)$ that
satisfies conditions \eqref{eq:p(n)} is a threshold for repair of depth $t$ in the ensemble $\cG_{n,p}.$
 
\subsection{Repair bandwidth}
In this section we estimate the communication complexity of  node recovery on a random graph. Throughout this section we assume that 
$t$ is the threshold for repair, i.e., conditions \eqref{eq:conditions} hold for $t,n,$ and $p$. Recall that by our assumption, $l=d-k+1.$ 
\begin{proposition}\label{lemma:bAF} The repair bandwidth $\beta_{\text{\rm AF}}$ satisfies
  $$
	\PP(\beta_{\text{\rm AF}} \ge td -o(n)) \rightarrow 1
  $$ where $t$ is the threshold for repair as given by \eqref{eq:conditions}.
\end{proposition}

\begin{proof} Rewriting the expression for $\beta_{\text{\rm AF}}$ in \eqref{eq:bAF}, we obtain
    $$
    \beta_{\text{\rm AF}}= td - \sum_{i=1}^{t-1}(t-i)|\Gamma_i(v_f)|.
    $$
Define the events $E_i = \{|\Gamma_i(v_f)| \le (1+\epsilon)(np)^i\}$ and notice that $E: = \cap_{i=1}^{t-1}E_i \subseteq \{\beta_{\text{\rm AF}} \ge (td -o(n))\}$. 
From Lemma \ref{lemma:Ni} we know that $\PP(E_i^c) \le 1/\log^2 n$ for all $i$, and thus
   $$
	Pr(\cup_{i=1}^{t-1}E_i^c) \le \sum_{i=1}^{t-1}Pr(E_i^c)\le t/{\log^2 n}.
	$$
Finally, $\PP(\beta_{\text{\rm AF}} \ge td-o(n)) \ge Pr(E) \ge 1-\frac{t}{\log^2 n}\rightarrow 1.$
\end{proof}

{\em Remark:} This proposition implies that for large $n,$ most of the helper nodes are at distance $t$ from the 
failed node. Note that, assuming \eqref{eq:conditions}, Lemma \ref{lemma:Ni} along with Lemma 8 in \cite{Chung01} imply that the size of the neighborhood $\Gamma_t(v)$ with high probability grows as $c(np)^t$ for some constant $c<1$. This provides an intuitive explanation of the claim of Prop.~\ref{lemma:bAF} for $d=\Theta(n)$ and $(np)^{t-1} = o(n)$, implying that the AF repair strategy results in a $t$-fold increase of repair bandwidth compared to full connectivity. 

\vspace*{.1in}The next proposition gives further insights into the relationship between $\beta_{\text{\rm AF}}$ and $t$. 
\begin{proposition}\label{lemma:kappa}
	Let $d=\delta n, 0<\delta<1,$ let $\kappa(n)$ be a function of $n$ such that $\underline{c}\le \kappa(n)/n\le \bar c$ for some 
constants $\underbar c, \bar c$ and all sufficiently large $n$,
and let $t$ be the threshold for repair as given by \eqref{eq:conditions}. We have
    $$
    \PP(\beta_{\text{\rm AF}} \le \kappa(n))\to
         \begin{cases} 
              0 &\text{if } t>\bar c/\delta\\
              1 &\text{ if } t\le \underbar c/\delta.
         \end{cases}     
    $$
\end{proposition}
\begin{proof}
To prove the first claim in the proposition, compute
   \begin{align*}
			\PP&(\beta_{\text{\rm AF}} \le \kappa(n)) 
			\le \PP(\beta_{\text{\rm AF}} \le \overline{c}n)\\
			&= \PP\Big(td - \sum_{i=1}^{t-1}(t-i)|\Gamma_i(v_f)| \le \overline{c}n\Big)\\
			&=\PP\Big(\sum_{i=1}^{t-1}(t-i)|\Gamma_i(v_f)| \ge (t\delta - \overline{c})n\Big)\\
			&\le \PP\Bigl(\sum_{i=1}^{t-1}(t-i)|\Gamma_i(v_f)| \ge (t\delta - \overline{c})n \Bigm\vert E\Bigr)\PP(E)+\PP(E^c),
	\end{align*}
where the event $E$ is defined above in Prop. \ref{lemma:bAF}. Conditional on $E$ we have 
$\sum_{i=1}^{t-1}(t-i)(1+\epsilon)(np)^i=\Theta((np)^{t-1})$, and \eqref{eq:conditions} implies that the first term $\to 0$ w.h.p. 
To complete the proof notice that $\PP(E^c)\le {t}/{\log^2 n}\to 0.$

For the second claim, observe that
 $$
\PP(\beta_{\rm AF} \le \kappa(n)) \ge \PP(\beta_{\text{\rm AF}} \le \underbar{c}n) \ge \PP(\beta_{\text{\rm AF}} \le td) = 1
  $$
  concluding the proof.
\end{proof}

Now let us show that the graph-based repair as defined in \eqref{eq:lc} or \eqref{eq:DM} with high probability has smaller repair bandwidth. %
\begin{theorem}\label{thm:IP}
	 Let $t$ be the threshold given in Prop.~\ref{lemma:threshold}. Let  $\chi(n)$ be such that 
   $
\chi(n)n^{s-1}p^s \rightarrow 0
  $
where $s \le t-1$ is the largest integer for which this condition holds. Then 
$
\PP(\beta_{\text{\rm IP}} \le (t-s)d+o(n)) \rightarrow 1.
$
\end{theorem}

{\em Remark:} Since $pn\to\infty$, it is easy to check that the assumptions of the theorem are non-vacuous, i.e., the largest $s$ satisfying
the condition exists and is well defined.

\begin{proof}
Let $T_f$ be the repair tree with the root $v_f$. By assumption, the distance from the root to the leaves is $t$, and we will assume that the
helper nodes in $\Gamma_i(v_f), i=t, t-1,\dots,s+1$ simply relay their information along the edges, while the nodes in $N_s(v_f)$ transmit
$l=d-k+1$ symbols given by a linear combination of the form given in \eqref{eq:lc}.

 Then, for the failed node $v_f$, we have
\begin{align*}
\beta_{\text{\rm IP}} &= (t-s)(d-|N_{t-1}(v_f)|)\\
&+\sum_{i=1}^{t-s-1}(t-s-i)|\Gamma_{t-i}(v_f)|
                 + (d-k+1)\sum_{i=1}^s|\Gamma_i(v_f)|\\
					&= (t-s)d+\sum_{i=1}^s|\Gamma_{i}(v_f)|(d-k+1-(t-s))\\
					& \hspace*{.3in}-\sum_{i=s+1}^{t-1}|\Gamma_{i}(v_f)|(t-i)\\
					&\le (t-s)d+\sum_{i=1}^s|\Gamma_{i}(v_f)|(d-k+1-t+s)).
\end{align*}
Proceeding similarly to the proof of Prop.~\ref{lemma:bAF}, we obtain
    \begin{multline*}
	\PP\Big(\beta_{\text{\rm IP}} \le (t-s)d+ \sum_{i=1}^s(1+\epsilon)(np)^{i}(\chi(n)+1-t+s)\Big)
	 \\\ge 1-{s}/{\log^2 n} \rightarrow 1.
    \end{multline*}
Now using the assumption $\chi(n)(np)^s = o(n)$ finishes the proof.
\end{proof}

To conclude, we have shown a strict separation between the typical communication cost of node recovery using the IP repair procedure and
the graph-agnostic AF protocol when the number of helpers $d$ is only slightly more than $k.$ Note the following simple corollary:

\begin{corollary}\label{cor_IP}
Let $t$ be the threshold given in Prop.~\ref{lemma:threshold}. For $\chi(n) = O(1)$, 
$$
\PP(\beta_{\text{\rm IP}} \le d+o(n)) \rightarrow 1.
$$
\end{corollary} 
\begin{proof}
	By \eqref{eq:conditions}, for $\chi(n) = O(1)$, the condition  $
	\chi(n)n^{s-1}p^s \rightarrow 0
	$ is satisfied for $s=t-1$. 
\end{proof}
Corollary \ref{cor_IP} supports the following intuition. Since $\chi(n)$ is a constant, nodes in all layers but the last can do a very high amount of compression and hence the contribution of those layers to the total communication complexity becomes insignificant; the complexity is primarily determined by the number of helper nodes in the last layer which is approximately $d$. This implies that even in random graphs where we do not have direct connectivity among all the helper nodes and the failed node, it is possible to bring down the communication complexity to the same order as that of the case of direct connectivity using IP. 

The above theorem and corollary suggest that the communication complexity is primarily controlled by the two parameters $p(n)$ and $\chi(n)$. One 
can ask the question, for what values of these parameters, does the complexity become significantly higher than that of the complexity of repair 
under full connectivity, i.e., $d$. 
In other words, we wish to study the behavior of $\beta^\ast-d$ where $\beta^\ast$ is the minimum complexity over all possible repair schemes\footnote{Our arguments rely on the information-theoretic lower bounds, so they indeed apply to all possible repair schemes.}. In this regard, Corollary \ref{cor_IP} says that for $\chi(n) = O(1)$, we have $\beta^\ast-d = o(n)$ with high probability. We will now show that for sparse graphs with high probability the repair becomes significantly more complex than sending $d$ symbols. 
The following theorem quantifies this claim. Its proof relies on Lemma \ref{lemma:bound} together with another lemma from \cite{Chung01}.
 \begin{lemma}[\!\!\cite{Chung01}, Lemma 2]\label{lemma:Ni_alt}
	Suppose that $p > \frac{c\log n}{n}$ for a constant $c \le 2$. Then with probability at least $1-o(\frac{1}{n})$, we have 
	for all $1\le i\le n$
	\begin{gather}
		|\Gamma_i(x)| \le \frac{9}{c}(np)^i \\
		|N_i(x)| \le \frac{10}{c}(np)^i . 
	\end{gather}	
\end{lemma}

\begin{theorem}\label{thm:sparse}
For $p(n) = o(\frac{\chi(n)}{n})$, we have $\PP(\beta^\ast -d = \Theta(n)) \rightarrow 1$.
\end{theorem}
\begin{proof}
Given $p(n)$, let $t$ be the threshold for repair as given in Prop.~\ref{lemma:threshold}. Clearly, 
any helper node $v\in N_{t}(v_f)$ needs to transmit at least one unit of information, so $\beta^\ast \ge d-|N_{t-1}(v_f)|$. Let $F_n := \{\Gamma_{t-1}(v_f) \le \frac{9}{c}(np)^{t-1}\}$. Now consider a node $v\in N_{t-1}(v_f)$. 
From Lemma \ref{lemma:Ni_alt}, the immediate neighborhood of this node satisfies
   $$
   \PP\Big(|\Gamma_{1}(v)| \le \frac{9}{c}np\Big) \ge 1-o\Big(\frac{1}{n}\Big)
   $$
for some constant $c \le 2$. Let $D(v)$ be the {\em immediate} descendants of node $v$ in the repair tree. For every $\delta>0$, 
there exists an $n_1$ such that $|D(v)| \le |\Gamma_{1}(v)| \le \frac{9}{c}np$ with probability at least $1-\frac{\delta}{n}$ for every $n \ge n_1$. 
Further, since $np = o(\chi(n)),$ for every $\epsilon > 0$, there exists an $n_2$ such that $np \le \epsilon\chi(n)$ for every $n \ge n_2$. 
Combining these two statements, we claim that the event $E_{v,\epsilon,n} :=\{|D(v)| \le \epsilon\chi(n)\}$ satisfies 
   $$
   \PP(E^c_{v,\epsilon,n})\le \delta/n \text{ for all }\epsilon,\delta >0, n\ge \max(n_1,n_2).
   $$
 Since this is true for all $v\in \Gamma_{t-1}(v_f),$ we have
    \begin{align*}
     \PP( \cup_{v \in \Gamma_{t-1}(v_f)} E_{v,\epsilon,n}^c) &=  \PP( \cup_{v \in \Gamma_{t-1}(v_f)} E_{v,\epsilon,n}^c| F_n)\PP(F_n)
     \\& \hspace*{.2in}+ 
     \PP( \cup_{v\in \Gamma_{t-1}(v_f)} E_{v,\epsilon,n}^c| F_n^c)\PP(F_n^c)\\
                  &\le \frac{9\delta}{c}\frac{(np)^{t-1}}{n}+ o\Big(\frac{1}{n}\Big) \rightarrow 0,
    \end{align*}
where the last step follows because by the definition of the threshold $t$, $(np)^{t-1} = o(n)$. This implies that $\PP(\cap_{v \in \Gamma_{t-1}(v_f)} E_{v,\epsilon,n}) \rightarrow 1$ for all $\epsilon>0$. 
Now by Lemma \ref{lemma:bound}, for $|D(v)|\le \chi(n)$ the outflow of communication from the set $D(v) \cup \{v\}$ has to be at least $|D(v)|+1$ and by the above analysis this is true for all $v \in \Gamma_{t-1}(v_f)$ with high probability. 
This implies that with high probability
   \begin{align*}
   \beta^\ast &\ge  d-|N_{t-1}(v_f)| + \sum_{v \in \Gamma_{t-1}(v_f)}(|D(v)|+1) \\
   &\ge  d-|N_{t-1}(v_f)| + \Gamma_{t}(v_f) = 2(d-|N_{t-1}(v_f)|).
   \end{align*}
Finally, noting that $\PP(|N_{t-1}(v_f)| = o(n)) \rightarrow 1$ gives the desired claim. 
\end{proof}

This theorem gives a sufficient condition for the separation of complexity of repair on a complete graph and a sparse random graph.

%
%

\subsection{Random Regular Graphs}
In this section we briefly address node repair on random regular graphs. We single out this ensemble from a multitude of other options
because it is conceivable that the architecture of the storage system places the same number of servers in close proximity to any single server,
and this is modeled by a regular graph. Let $\cG_{n,r}$ be the set of all $r$-regular $n$-vertex graphs with a uniform distribution on it. We denote 
a random element from this ensemble by $\bG_{n,r}$.  Assume again that the data is encoded with an $(n,k,d,l)$ MSR code, and the coordinates
$C_1,\dots,C_n$ of the codeword are placed on the vertices  $v_1,\dots,v_n.$ 

We will derive conditions on the parameters $k=k(n),d=d(n)$ and $r=r(n)$ such that as $n \rightarrow \infty,$ 
with high probability the graph-based repair process is advantageous over the AF strategy. We again assume that $d=\Theta(n).$
Denote by $\cG_{n,m}$ the ensemble of graphs with $n$ vertices and $m$ edges and let $\bG_{n,m}$ be a random graph sampled from it.

For the purposes of node repair we need the graph to be connected. In \cite{Bollobas80}, Bollob{\'a}s showed that $\bG_{n,r}$ is $r$-connected with high probability.

Recall that a property of graphs is called {\em increasing} if it is inherited from a subgraph to any graph that contains it. The following equivalence between properties of $\bG_{n,m}$ and $\bG_{n,r}$ will be used below.

\begin{lemma}\label{lemma:equiv}(\cite{FK2016}, Corollary 10.11)
Let $\cL$ be an increasing property of graphs such that $\bG_{n,m}$ satisfies $\cL$ with high probability for some $m=m(n),$ where $n\log n \ll m \ll n^2$. Then $\bG_{n,r}$ satisfies $\cL$ with high probability for $r=r(n) \sim \frac{2m}{n}$.
\end{lemma}
The following proposition is a counterpart to Prop.~\ref{lemma:threshold} for random regular graphs. Here the definition of the threshold depth of
repair is the same as in Sec.~\ref{sec:threshold}.
\begin{proposition}\label{lemma:threshold_reg}
	Let $d=\delta n, 0<\delta<1$ and let $t$ be a fixed integer. Then $t$ is the threshold depth for repair if
	\begin{gather}\label{eq:conditions_reg}
		r^{t-1}=o(n), \quad
		\frac{r^t}{n}-2\log n \rightarrow \infty.
	\end{gather}
\end{proposition}
\begin{proof}
For finite $t$, $|N_{t-1}(v)| \le \sum_{i=1}^{t-1}r^i = \Theta(r^{t-1}) = o(n) \ll d$ and so $(t-1)$-layer repair is not possible.  

For the other direction, let $r(n)$ satisfy relations \eqref{eq:conditions_reg}. Take $p(n) = \frac{r(n)}{n}$, then by Prop.~\ref{lemma:threshold} $\mathbb{G}_{n,p}$ satisfies $t$-layer repair with high probability. Recall a basic fact that the ensembles
$\cG_{n,p}$ and $\cG_{n,m}$ are equivalent for all monotone properties (i.e., graphs sampled from them either both have the
property w.h.p. or they both do not). Therefore, the graph $\mathbb{G}_{n,m}$ with $m(n) \sim \frac{n^2}{2}p(n) = \frac{nr(n)}{2}$ 
affords $t$-layer repair with high probability. Now, satisfying $t$-layer repair is a monotone increasing graph property and so by Lemma \ref{lemma:equiv} we have that $\mathbb{G}_{n,r}$ affords $t$-layer repair with high probability. 
\end{proof}
With the threshold conditions established, counterparts of Prop.~\ref{lemma:bAF}, \ref{lemma:kappa} and Theorem \ref{thm:IP} can 
be proved for $\bG_{n,r}$ by simply replacing $np$ with $r$ and proceeding along similar arguments (which are in fact simpler
because the neighborhood sizes of a vertex afford uniform bounds). In the following statements, given without 
proofs, $t$ is the repair threshold as defined in \eqref{eq:conditions_reg}.

\begin{proposition}\label{lemma:bAF_reg} The repair bandwidth $\beta_{\text{\rm AF}}$ satisfies
	$$
	\PP(\beta_{\text{\rm AF}} \ge td -o(n)) \rightarrow 1.
	$$ 
\end{proposition}
\begin{proposition}\label{lemma:kappa_reg}
	Let $d=\delta n, 0<\delta<1$ and let $\kappa(n)$ be a function of $n$ such that $\underline{c}\le \kappa(n)/n\le \bar c$ starting with some $n$. Then 
	$$
	\PP(\beta_{\text{\rm AF}} \le \kappa(n))\to
	   \begin{cases} 
	     0&\text{if }t>\bar c/\delta\\
	     1&\text{if }t\le \underbar c/\delta.
	   \end{cases}
	$$
	for $t>\bar c/\delta$  we have $\PP(\beta_{\text{\rm AF}} \le \kappa(n)) \rightarrow 0.$
\end{proposition}
\begin{theorem}\label{thm:IP_reg}
	 Let $d-k = \chi(n)$ be a function of $n$ such that 
	$
	\frac{\chi(n)r(n)^s}{n} \rightarrow 0
	$
	where $s \le t-1$ is the largest integer for which this condition holds. Then 
	$
	\PP(\beta_{\text{\rm IP}} \le (t-s)d+o(n)) \rightarrow 1.
	$
\end{theorem}
Similarly to the case of $\cG_{n,p},$ this shows a strict separation between the typical communication cost of node recovery using the IP repair procedure and the graph-agnostic AF protocol under a certain assumption on $\chi(n).$ A theorem that parallels Theorem \ref{thm:sparse} can be also easily established (note that the bounds of the form given in Lemma~\ref{lemma:Ni_alt} for regular 
graphs come for free).

\section{Conclusion}
In this paper we posed and advanced the problem of erasure correction (node repair) when the elements of encoded information are placed on the nodes of a graph, adopting the total amount of communication for repair as the figure of merit in the analysis. The main
difference of this problem from the standard setting of regenerating codes stems from the fact that most helpers are not directly connected to the failed node, and the information transmitted by them can be processed by the intermediate nodes or combined with the contents
of these nodes. We showed that the intermediate processing scheme can be implemented by linear MSR codes, attaining the general lower bounds on complexity derived in the paper. These results were also extended to the case of multiple failed nodes. We also established a framework for the analysis of repair schemes when the helpers communicate among themselves before contributing data toward the repair task, and gave simple examples when the arising complexity bounds are attained with equality. Finally, we studied the repair problem when the underlying graph is random, establishing bounds on the edge probability under which the intermediate processing scheme 
provides complexity savings compared to simple relaying.

Among the problems that so far have resisted analysis is the case when some of the helper nodes provide incorrect information. Since
the erroneous information is propagated along the edges and potentially combined with the contents of the intermediate nodes, the repair procedures proposed
in this paper do not support node recovery. This is unlike the case of complete connectivity, where optimal 
repair is possible even in the presence of errors \cite{Ye17}. Another version of repair with errors assumes that the edges form noisy channels, so the information propagated along them is received as realizations of random variables. A simple way to address this
problem suggests to add redundancy to the data transmitted over the edges, and it combines channel coding and the repair task.
Optimizing the tradeoffs that arise as a result presents an open problem.
Yet another challenge is to construct codes and nontrivial repair protocols when the helper nodes communicate among themselves in the process of repair, extending the approach of Sec.~\ref{sec:intra}. 

\appendix
We present an encoding/repair scheme for the example in Sec.~\ref{sec:coop}, Fig.3 that enables recovery of the contents of the 
node $v_f$ which assumes that the nodes in $\Gamma_1(v_f)$ exchange information before passing the repair data to the node $v_f.$ The 
construction shares some features of cooperative repair, and it relies on a code family constructed earlier for the case of the complete 
graph $K_n$ \cite[Sec.IV]{Ye17}. Namely, suppose that the information is encoded with an $[n,k,d=k+1,l=2^n]$ MSR array code $\cC$ 
and the codeword coordinates are placed on the vertices of a graph $G(V,E)$ with $|V|=n.$ Suppose further that the repair graph of the failed node $v_f$ is as shown in the figure. 
In accordance with \cite{Ye17}, Construction 2, we will assume that each helper node provides $\beta = 2^{n-1}$ symbols for the repair of $v_f.$

Let $C=(c_1,c_2,\dots,c_n)\in \cC$ be a codeword with $c_i=(c_{i,0},c_{i,1},\dots,c_{i,l-1}) \in \mathbb{\ff}^l$. The code is defined by the
following parity-check equations:
    \begin{multline*}
	\sum_{i=1}^n\lambda_{i,a_i}^tc_{i,a} = 0 \;\;\;\;\text{for all}\; a \in \{0,1,\dots,l-1\},\\ t\in \{0,1,\dots,n-k-1\},
  \end{multline*}
where $(a_1,a_2,\dots,a_n)$ is the binary representation of $a$. Below we assume that $v_f=1$, that $\Gamma_1(1)=\{2,3\}.$ and that
$\Gamma_2(1)=\{4,5,\dots,k+2\}.$ For a string $s=(s_1,s_2,\dots,s_{n-k-2}) \in \{0,1\}^{n-k-2}$, consider the set of $2^{k+2}$ values of $a$ for which $(a_{k+3},a_{k+4},\dots,a_n) = s$. Isolate a subset of this set by fixing a string $\hat{s}=(\hat{s}_1,\hat{s}_2,\dots,\hat{s}_{k-1}) \in \{0,1\}^{k-1}$ and collecting only those values of $a$ for which $(a_4,a_5,\dots,a_{k+2}) = \hat{s}$. 
Having fixed $s$ and $\hat{s}$, we are left with 8 parity check equations which can be labeled 
by a binary vector $\tilde{s} \in \{0,1\}^3$:
   \begin{multline}\label{eqtn_pc}
		\lambda_{1,\tilde{s}_1}^tc_{1,(\tilde{s},\hat{s},s)}+\lambda_{2,\tilde{s}_2}^tc_{2,(\tilde{s},\hat{s},s)}+\lambda_{3,\tilde{s}_3}^tc_{3,(\tilde{s},\hat{s},s)}\\ + \sum_{i=4}^{k+2}\lambda_{i,\hat{s}_{i-3}}^tc_{i,(\tilde{s},\hat{s},s)}+\sum_{i=k+3}^{n}\lambda_{i,s_{i-k-2}}^tc_{i,(\tilde{s},\hat{s},s)}=0
		\\\forall\;\;\tilde{s} \in \{0,1\}^3,\;\;t \in \{0,1,\dots,n-k-1\}
    \end{multline}	
For fixed $\hat{s}$ and $s$, the $\lambda$'s in the last two sums the same in all the equations. Define
\begin{align*}
	\mu^{(\hat{s},s)}_{2,1,i} &= c_{i,(000,\hat{s},s)}+c_{i,(010,\hat{s},s)}+c_{i,(100,\hat{s},s)}\\
	\mu^{(\hat{s},s)}_{3,1,i} &= c_{i,(000,\hat{s},s)}+c_{i,(100,\hat{s},s)}+c_{i,(101,\hat{s},s)}\\
	\mu^{(\hat{s},s)}_{2,2,i} &= c_{i,(001,\hat{s},s)}+c_{i,(011,\hat{s},s)}+c_{i,(111,\hat{s},s)}\\
	\mu^{(\hat{s},s)}_{3,2,i} &= c_{i,(011,\hat{s},s)}+c_{i,(110,\hat{s},s)}+c_{i,(111,\hat{s},s)}
\end{align*}
For $i \in \{4,5,\dots,k+2\}$ the helper node $i$  sends $\mu^{(\hat{s},s)}_{2,1,i},\mu^{(\hat{s},s)}_{2,2,i}$ to node 2 and $\mu^{(\hat{s},s)}_{3,1,i},\mu^{(\hat{s},s)}_{3,2,i}$ to node 3. Additionally node 3 sends $\mu^{(\hat{s},s)}_{2,1,3},\mu^{(\hat{s},s)}_{2,2,3}$ to node 2 and node 2 sends $\mu^{(\hat{s},s)}_{3,1,2},\mu^{(\hat{s},s)}_{3,2,2}$ to node 3. 

Node 2, having $\mu^{(\hat{s},s)}_{2,1,i}$ for all $i \in \{3,\dots,k+2\}$, can recover $c_{1,(000,\hat{s},s)}+c_{1,(010,\hat{s},s)}, c_{1,(100,\hat{s},s)}$. To see this, sum Eqns.~\eqref{eqtn_pc} for $\tilde{s}\in\{ 000,010,100\}$ to obtain
\begin{multline*}	
		\lambda_{1,0}^t(c_{1,(000,\hat{s},s)}+c_{1,(010,\hat{s},s)})+\lambda_{1,1}^tc_{1,(100,\hat{s},s)}
		\\+\lambda_{2,0}^t(c_{2,(000,\hat{s},s)}+c_{2,(100,\hat{s},s)})+\lambda_{2,1}^tc_{2,(010,\hat{s},s)}
		\\+ \lambda_{3,0}^t\mu^{(\hat{s},s)}_{2,1,3} + \sum_{i=4}^{k+2}\lambda_{i,\hat{s}_{i-3}}^t\mu^{(\hat{s},s)}_{2,1,i}+\sum_{i=k+3}^{n}\lambda_{i,s_{i-k-2}}^t\mu^{(\hat{s},s)}_{2,1,i}=0.
\end{multline*}
The multiplies of the $\lambda$'s in this equation form a codeword
of an $(n+2,k+2 = d+1)$ Reed-Solomon code. Node 2 collects $\mu^{(\hat{s},s)}_{2,1,i}$ for all $i \in \{3,\dots,k+2\}$ and it already knows $(c_{2,(000,\hat{s},s)}+c_{2,(100,\hat{s},s)}), c_{2,(010,\hat{s},s)}$, and so it can recover $(c_{1,(000,\hat{s},s)}+c_{1,(010,\hat{s},s)})$ and $c_{1,(100,\hat{s},s)}$. Similarly, it can be shown that with $\mu^{(\hat{s},s)}_{2,2,i}$ for all $i \in \{3,\dots,k+2\}$, node 2 can recover $(c_{1,(001,\hat{s},s)}+c_{1,(011,\hat{s},s)})$ and $c_{1,(111,\hat{s},s)}$. 

Node 3, using $\mu^{(\hat{s},s)}_{3,1,i}$ and $\mu^{(\hat{s},s)}_{3,2,i}$ for all $i \in \{2,4,5\dots,k+2\},$ can recover $(c_{1,(100,\hat{s},s)}+c_{1,(101,\hat{s},s)})$ and $c_{1,(000,\hat{s},s)}$ and $(c_{1,(110,\hat{s},s)}+c_{1,(111,\hat{s},s)})$ and $c_{1,(011,\hat{s},s)}$. Nodes 2 and 3 send these recovered linear combinations to the failed node, which can recover all $c_{1,(\tilde{s},\hat{s},s)}$ for $\tilde{s} \in \{0,1\}^3$. Finally this is done for all $\hat{s}$ and $s,$ and this recovers the entire of node 1.

{\em Communication Complexity:} Each helper node $i \in \{4,5,\dots,k+2\}$ sends two symbols to node 2 and 2 symbols to node 3 for each fixed $\hat{s},s$. Hence they transmit 4 symbols each resulting in a total transmission of $4(k-1)$ for each fixed $\hat{s},s$. Similarly node 3 transmits 2 symbols to node 2 and node 2 transmits 2 symbols to node 3 for each fixed $\hat{s},s$. Finally node 2 and node 3 total transmit 8 symbols to node 1 for each fixed $\hat{s},s$. This is repeated for every possible $\hat{s}$ and $s$. Hence the total communication complexity is 
\begin{align*}
	B &= (4(k-1)+4+8)\cdot 2^{k-1}\cdot 2^{n-k-2} 
	\\&= 2^{n-1}k + 2^n = \beta(d-1)+l=(d+1)\beta
\end{align*}
(since $l=2\beta$).
This matches the communication complexity of repair attainable with an IP protocol of Sec.~\ref{sec:coop}.

\remove{\vspace*{.2in}
{\bf B.} {\em Proof of Lemma \ref{lemma:bound_mult}:}
	1) By the assumption \eqref{eq:eqtn_mult1}, given the contents of all the nodes in $D\backslash A,$ the information contained in $R_A^F$ is sufficient to repair the nodes $\{v_i: i\in F\}$, i.e., 
	\begin{equation}\label{eq:eqtn_mult2}
		H(W_F|R_A^F, W_{D\backslash A})=0.
	\end{equation}
	We have $|D\backslash A| \le k-h$. Consider a set $B \subset A$ with $|B| = k-h-|D\backslash A|$. Now,
	\begin{equation}\label{eq:eqtn_mult3}
		H(R_A^F, W_{D\backslash A}, W_{B}) = H(R_A^F, W_{D\backslash A}, W_F, W_{B}) \ge kl
	\end{equation}
	where the equality in \eqref{eq:eqtn_mult3} follows from \eqref{eq:eqtn_mult2} and the chain rule, and the inequality follows from the MDS property of MSR codes because 
	$|D\backslash A|+|B|+|F| = k$. Next observe that
	\begin{align}
		H(R_A^F, W_{ D\backslash A}, W_{B}) &\le H(R_A^F)+H( W_{D\backslash A}, W_{B})\notag\\ &= H(R_A^f)+ (k-h)l\label{eq:eqtn_mult4}
	\end{align}
	where the equality again uses the independence of any $k-h$ coordinates in an MDS code. Combining \eqref{eq:eqtn_mult3} and \eqref{eq:eqtn_mult4}, we obtain the claimed inequality.
	
	For Part (2), let $C \subseteq D\backslash A$ such that $|C| = k-h$ and let $I = D\backslash \{A\cup C\}$. By the assumption \eqref{eq:eqtn_mult1}, we have
	\begin{equation}\label{eq:eqtn_mult5}
		H(W_F|R_A^F,W_C,S_I^F) = 0.
	\end{equation}
	Now,
	\begin{equation}\label{eq:eqtn_mult6}
		H(R_A^F,W_C,S_I^F) = H(R_A^F,W_F,W_C,S_I^F) \ge kl
	\end{equation}
	where the equality in \eqref{eq:eqtn_mult6} follows from \eqref{eq:eqtn_mult5} and the chain rule, and the inequality follows from the MDS property and the fact that $|C| = k-h$. Next observe that
	\begin{equation}\label{eq:eqtn_mult7}
		\begin{split}
			&H(R_A^F,W_C,S_I^F)\\ &\le H(R_A^F)+H( W_C)+H(S_I^F)\\ 
			&\le H(R_A^f)+H( W_C)+\sum_{i \in D\backslash\{A\cup C\}}H(S_i^F)\\
			&= H(R_A^F)+ (k-h)l+\frac{h(d-(k-h)-|A|)l}{d-k+h}
		\end{split}
	\end{equation}
	where we again use the independence of any $k-h$ coordinates in an MDS code. Combining \eqref{eq:eqtn_mult6} and \eqref{eq:eqtn_mult7}, we obtain the claimed inequality.}
\qed


\end{document}